\documentclass[12pt]{article}
\usepackage[utf8]{inputenc}

\usepackage{amsmath}
\usepackage{amsfonts}
\usepackage{amssymb}
\usepackage{anysize}
\usepackage{cite}
\usepackage{comment}
\usepackage[shortlabels]{enumitem}
\usepackage{mathrsfs}
\usepackage{amsthm}
\usepackage{tikz}
    \usetikzlibrary{matrix}
    \usetikzlibrary{decorations.markings}
\usepackage{hyperref}
\usepackage[T1]{fontenc}
\usepackage{oldgerm}
\usepackage{scalerel,stackengine}
\usepackage{stmaryrd}
\usepackage{xcolor}
\usepackage{xfrac}
\usepackage{authblk}

\usepackage{framed}

\stackMath
\newcommand\reallywidehat[1]{%
\savestack{\tmpbox}{\stretchto{%
  \scaleto{%
    \scalerel*[\widthof{\ensuremath{#1}}]{\kern.1pt\mathchar"0362\kern.1pt}%
    {\rule{0ex}{\textheight}}
  }{\textheight}%
}{2.1ex}}%
\stackon[-6.9pt]{#1}{\tmpbox}%
}
\newcommand{\ad}{\operatorname{ad}}
\newcommand{\Ad}{\operatorname{Ad}}

\newtheorem{theorem}{Theorem}[section]
\newtheorem{definition}{Definition}[section]

\newtheorem{lemma}{Lemma}[section]

\newtheorem{remark}{Remark}[section]

\title{Nonholonomic reduction for mechanical systems with collisions}

\author[1]{Álvaro Rodríguez Abella\thanks{alvrod06@ucm.es}}

\author[2]{Leonardo J. Colombo\thanks{leonardo.colombo@csic.es}}

\affil[1]{Department of Mathematics and Computer Science, Saint Louis University (Madrid Campus), Avenida del Valle, 34, Madrid, 28003, Madrid, Spain}

\affil[2]{Centro de Automática y Robótica (CSIC-UPM), Carretera de Campo Real, km 0, 200, 28500, Arganda del Rey, Spain.}

\date{}

\begin{document}

\maketitle

\begin{abstract}
This paper studies nonsmooth variational problems on principal bundles for nonholonomic systems with collisions taking place in the boundary of the manifold configuration space of the nonholonopmic system. In particular, we first extended to a nonsmooth context appropriate for collisions the variational principle for nonholonomic implicit Lagrangian systems, to obtain implicit Lagrange--d'Alembert--Pontryagin equations for nonholonomic systems with collisions, and after introducing the notion of connection on a principal bundle we consider Lagrange--Poincar\'e--Pointryagin reduction by symmetries for systems with collisions. 
\end{abstract}



\section{Introduction}

 Reduction theory is one of the fundamental tools in the study of mechanical systems
with symmetries and it essentially concerns the removal of certain variables by using the symmetries of the system and the associated
conservation laws. Such symmetries arise when one has a  Lagrangian that is  invariant
under a Lie group action $G$, i.e., when the Lagrangian function is invariant under the tangent lift of the action of the Lie group on the configuration manifold $Q$. If we denote by $\Phi_g:Q\to Q$ this (left) action, $g\in G$, then the invariance condition under the tangent lift action is expressed by $L\circ T\Phi_g=L$. If such an invariance property holds when the configuration manifold is the group itself, $Q=G$, and the action is given by left translations, that is, $\Phi_g=L_g:G\to G$ is given by $L_g(h)=gh$ for each $h\in G$, we say that the Lagrangian $L$ is $G$-invariant. For
a symmetric mechanical system, reduction by symmetries eliminates the directions
along the group variables and thus provides a system with fewer degrees of freedom. 

If the (finite-dimensional) differentiable
manifold $Q$ has local coordinates $(q^i)$, $1\leq i\leq$ dim$\,Q$ and we denote by $TQ$ its
tangent bundle with induced local coordinates $(q^i, \dot{q}^i)$, given a Lagrangian function $L:TQ\rightarrow \mathbb{R}$, its Euler--Lagrange
equations are
\begin{equation}\label{qwer}
\frac{d}{dt}\left(\frac{\partial L}{\partial\dot
q^i}\right)-\frac{\partial L}{\partial q^i}=0, \quad 1\leq i\leq \mbox{dim}\,Q.
\end{equation} As is well-known, when $Q$ is the configuration manifold  of a mechanical system, equations \eqref{qwer} determine its dynamics.

A paradigmatic example of reduction is the derivation of the Euler--Poincar\'e equations from the Euler--Lagrange equations  \eqref{qwer} when the configuration manifold is a Lie group, i.e. $Q=G.$ Assuming that the Lagrangian $L:TG\to\mathbb{R}$ is left invariant under the action of $G$ it is possible to reduce the system by introducing the body fixed velocity $\xi\in\mathfrak{g}$  and the reduced Lagrangian $\ell:(TG)/G\simeq\mathfrak{g}\to\mathbb{R}$, provided by the invariance condition $\ell(\xi)=L(g^{-1}g,g^{-1}\dot{g})=L(e,\xi)$. The dynamics of the reduced Lagrangian is governed by the \textit{Euler--Poincar\'e} equations (see \cite{bloch1996euler} and \cite{holm1998euler} for instance) and given by the system of first order ordinary differential equations
\begin{equation}\label{Euler--Poincare-Eq}
\frac{d}{dt}\left(\frac{\delta\ell}{\delta\xi}\right)=\ad^{*}_{\xi}\left(\frac{\delta\ell}{\delta\xi}\right).
\end{equation}
This system, together with the reconstruction equation $\xi(t)=g^{-1}(t)\dot{g}(t)$, is equivalent to the Euler--Lagrange equations on $G$, which are given by
\[
\frac{d}{dt}\left(\frac{\partial L}{\partial\dot
g}\right)=\frac{\partial L}{\partial g}.
\] 

Reduction theory for mechanical systems with symmetries can be also developed by
using a variational principle formulated on a
principal bundle $\pi_{Q,\Sigma}:Q \to\Sigma$, where $\Sigma= Q/G$ and a principal connection $\omega$ is introduced on $Q$ in order to define a bundle isomorphism (cf. \cite{CeMaRa2001})
\[
(TQ)/G \rightarrow T\Sigma \oplus \tilde{\mathfrak {g}},\quad
[v_q]_G\mapsto\left( T_q\pi_{Q,\Sigma} (v_q), [ q, \omega_q(v_q)]_G\right),
\] where $\oplus$ denotes the fibered direct sum over $\Sigma$, the bracket is the standard Lie bracket on the Lie algebra $\mathfrak{g}$ and $\tilde{\mathfrak{g}}=(Q\times\mathfrak{g})/G$ is the adjoint bundle of $\pi_{Q,\Sigma}$. A curve $q:[t_0,t_1]\to Q$ induces the two curves
$\sigma=\pi_{Q,\Sigma}\circ q:[t_0,t_1]\to\Sigma$ and $\overline\zeta=[q, \omega (\dot q)]_G:[t_0,t_1]\to\tilde{\mathfrak{g}}$. Variational Lagrangian reduction \cite{CeMaRa2001} states that the
Euler--Lagrange equations on $Q$ for a $G$-invariant Lagrangian $L$
are equivalent to the Lagrange--Poincar\'e equations on $(TQ)/G \simeq T\Sigma \oplus  \tilde{\mathfrak{g}}$ for the reduced
Lagrangian $\ell :T\Sigma \oplus  \tilde{\mathfrak{g}}\to\mathbb{R}$, which read
\begin{equation}\label{Intro-LagPoincareGeneral}
\left\{
\begin{array}{l}
\displaystyle\vspace{0.2cm}\frac{\nabla^{\omega*}}{dt}\frac{\delta\ell }{\delta\overline\zeta} - \mbox{ad}^*_{\overline\zeta}\frac{\delta\ell }{\delta\overline\zeta}=0
,\\
\displaystyle\frac{\delta\ell }{\delta\sigma} - \frac{\nabla^{\Sigma*}}{dt}\frac{\delta\ell }{\delta\dot{\sigma}} = \frac{\delta\ell}{\delta\overline\zeta}\cdot\left(i_{\dot{\sigma}}\tilde F^\omega\right),
\end{array}
\right.
\end{equation} where $\tilde F^\omega$ is the reduced curvature form associated
to the principal connection $\omega$, $\nabla^{\omega*}/dt$ denotes the 
covariant derivative in the coadjoint bundle bundle, and $\nabla^{\Sigma*}/dt$ denotes the covariant derivative on the cotangent bundle.

Some mechanical systems have a restriction on the configurations or velocities that the system may assume. Systems with such restrictions are generally called constrained systems. Nonholonomic  systems  \cite{bloch2003nonholonomic}, \cite{ne_mark2004dynamics} are, roughly speaking, mechanical systems with constraints on their velocity that are not derivable from position constraints. They arise, for instance, in mechanical systems that have rolling contact (e.g., the rolling of wheels without slipping) or certain kinds of sliding contact (such as the sliding of skates).  There are some differences between nonholonomic systems and classical Hamiltonian or Lagrangian systems. Among them, nonholonomic systems are nonvariational, they arise from Lagrange-d’Alembert principle and not from Hamilton’s principle; they may preserve the energy of the system as Hamiltonian systems, but they are not, in general, time-reversible, and they do not preserve, in general, the momentum for systems with symmetries (i.e., Noether’s theorem does not apply, in general, for nonholonomic systems). 
In this paper we study the reduction by symmetries of implicit nonholonomic systems, that is, nonholonomic dynamics obtained through the reduction of Hamilton-Pontryagin principle giving rise to a set of second-order implicit ordinary differential equations. Redcution by symmetries for nonholonomic systems has been developed in \cite{cendra2006lagrange}, \cite{cendra2008lagrangian} and \cite{cendra2001geometric}, and implicit nonholonomic systems in \cite{GaYo2015}.

Mechanical systems subject to collisions are confined within a region of space with boundary. Collision with the boundary for elastic impacts activates constraint on the momentum and energy after and before the collision occurs. The problem of collisions has been extensively treated in the
literature since the early days of mechanics (see \cite{brogliato1999nonsmooth} for a comprenshive review and references therein). More recently, much work has been done on the rigorous mathematical foundation of impact problems \cite{haddad2006impulsive}, \cite{westervelt2018feedback} for applications to bipedal locomotion. Nonholonomic systems subject to impacts or impulse effects has been previously studied in \cite{clark2019bouncing}, \cite{colombo2022geometric}. In terms of reduction by symmetries for systems with collisions/impulse effects, a hybrid scheme for Routh reduction for hybrid Lagrangian systems with cyclic variables is found in~\cite{ames2006hybrid} and \cite{colombo2020symmetries}, inspired to gain a better understanding of bipedal walking models (see also~\cite{ames2007geometric} and references therein). Symplectic reduction for hybrid Hamiltonian systems has been introduced in \cite{ames2006hybrid2} and extended to time-dependent systems in \cite{eyrea2020note}. Poisson reduction of hybrid Hamiltonian systems has been studied in \cite{EyCoBl2021}.  The main goal of this paper is to provide a theory for reduction by symmetries of nonholonomic systems subject to collision conditions from a non-smooth mechanics point of view by using techniques of variational calculus on manifolds and the theory of connections on principal bundles closing the gap of reduction theories for mechanical systems subject to collisions.


The remainder of the paper is structured as follows. Section \ref{sec:main} introduces nonholonomic systems. Section \ref{sec3} studies nonholonomic implicit Euler--Lagrange equations with collisions via Hamilton–d’Alembert–Pontryagin principle for non-smooth systems. In Section \ref{sec4} we study nonholonomic implicit Lagrange–Poincar\'e reduction with collisions, in particular, we study the case of reduced nonholonomic systems on Lie algebras and Euler--Poincar\'e--Suslov systems. Applicability examples are shown throughout the entire paper.

\section{Implicit nonholonomic systems}
\label{sec:main}
Let $Q$ be a differentiable manifold with $\hbox{dim}(Q)=n$, and $q^{i}$ be a particular choice of local coordinates on this manifold. In the following, $TQ$ denotes the tangent bundle of $Q$, with $T_{q}Q$ being the tangent space at each point $q\in Q$ and $v_{q}\in T_q Q$ being a vector. In addition, the coordinate chart $q^{i}$ induces a natural coordinate chart on $TQ$ denoted by $(q^{i},\dot{q}^{i})$. Let $T^{*}Q$ be the cotangent bundle of $Q$, which is locally described by the positions and the momenta of the system, i.e., $(q,p)\in T^{*}Q$. The cotangent space at each point $q\in Q$ is denoted by $T_{q}^{*}Q$. 
 

 A $k$-dimensional distribution on $Q$ is a vector subbundle $\Delta_Q\subset TQ$ with $k$-dimensional fiber, i.e., $\Delta_Q(q)\subset T_q Q$ is a $k$-dimensional subspace for each $q\in Q$. Moreover, $\Delta_Q$ is smooth if there exist a neighborhood $U$ of each point $q\in Q$ and local vector fields $X_1,\ldots,X_k\in\mathfrak X(U)$ that span $\Delta_Q$ on $U$, that is, $\Delta_Q(q)=\hbox{span}\{X_{1}(q),\ldots,X_k(q)\}$ for all $q\in U$. 
define \textit{codistributions.} 

Analogously, a $k$-dimensional codistribution on
$Q$ is a vector subbundle $\widetilde{\Delta_Q}\subset aT^{*}Q$ with $k$-dimensional fiber. Given the concept of codistribution, it is possible to define the \textit{annihilator} of a distribution $\Delta_Q\subset TQ$; namely, it is the codistribution given by $\Delta_Q^\circ(q)=\{\alpha\in T_{q}^{*}Q\,\,|\,\,\langle\alpha,v\rangle=0,\,\,\forall v\in\Delta_Q(q)\},\, q\in Q$, where $\langle\cdot,\cdot\rangle$ denotes the dual pairing.

Linear constraints on the
velocities are locally given by equations of the form
$\phi^{a}(q^i, \dot{q}^i)=\mu^a_i(q)\dot{q}^i=0, \, 1\leq a\leq
m$, depending, in general, on their configuration coordinates and  their
velocities. {}From an intrinsic point of view, the linear
constraints are defined by a distribution $\Delta_Q$ on
$Q$ of constant rank $n-m$ such that the annihilator of $\Delta_Q$ is locally given at each point of $Q$ by $\Delta_Q^\circ(q) = \operatorname{span}\left\{ \mu^{a}(q)=\mu_i^{a}dq^i \; \mid 1 \leq a
\leq m \right\}$, where the $1$-forms $\mu^{a}$ are linearly independent at each point of $Q$. When the constraint distribution $\Delta_Q$ is nonintegrable, the linear constraints are said to be nonholonomic. 

In addition to these constraints, we
need to specify the dynamical evolution of the system, usually by
fixing a Lagrangian function $L\colon  TQ \to \mathbb{R}$. The
central concepts permitting the extension of mechanics from the
Newtonian point of view to the Lagrangian one are the notions of
virtual displacements and virtual work. These concepts were
formulated in the developments of mechanics and in their application to
statics. In nonholonomic dynamics,  the procedure is given by the
\textit{Lagrange--d'Alembert principle}.
 This principle allows us to determine the set of possible values of the constraint forces from the set $\Delta_Q$ of admissible kinematic states alone. The resulting equations of motion are
\begin{equation*}
\left[ \frac{d}{dt}\left( \frac{\partial L}{\partial \dot
q^i}\right) - \frac{\partial L}{\partial q^i} \right] \delta
q^i=0,
\end{equation*}
where $\delta q^i$ denotes the virtual displacements verifying $\mu^a_i\delta q^i =0$.
By using Lagrange multipliers, we
obtain
\begin{equation}\label{ldaeq}
\frac{d}{dt}\left( \frac{\partial L}{\partial \dot
q^i}\right)-\frac{\partial L}{\partial q^i}={\lambda}_a\mu^a_i  .
\end{equation}
The term on the right-hand side represents the  constraint force or
reaction force induced by the constraints and the functions $\lambda_a$ are the Lagrange multipliers which, after
being computed using the constraint equations, allow us to obtain a
set of second-order differential equations.

Alternatively to the use of Lagrange multipliers, the phase space may be enlarged to the Pontryagin bundle $TQ\oplus T^{*}Q$ and the Lagrange--d'Alembert--Pontryagin principle may be considered. This variational principle is given by $$\delta\int_{t_0}^{t_1}\left(L(q(t),v(t))+\langle p(t),\dot{q}(t)-v(t)\rangle\right)dt=0,$$
where $v(t)\in\Delta_Q(q(t))$ and the variations $(\delta q(t),\delta v(t), \delta p(t))$ are such that $\delta q(t)\in\Delta_Q(q(t))$ and vanishes at the endpoints. Then stationary condition for a curve $(q(t),v(t),p(t))$ yields the implicit Lagrange--d'Alembert equations on $TQ\oplus T^{*}Q$ (see \cite{yoshimura2006dirac}): $$p=\frac{\partial L}{\partial v},\quad\dot{q}=v\in\Delta_Q(q),\quad\dot{p}-\frac{\partial L}{\partial q}\in\Delta_Q^\circ(q).$$

Now, let $G$ be a finite dimensional Lie group. A \textit{left action} of $G$ on a manifold $Q$ is a smooth mapping $\Phi:G\times Q\to Q$ such that $\Phi(e,q)=q$ for all $q\in Q$, where $e\in G$ denotes the identity element, and $\Phi(g,\Phi(h,q))=\Phi(gh,q)$ for all  $g,h\in G, q\in Q$. In particular, for each $g\in G$, the map $\Phi_g:Q\to Q$ defined as $\Phi_{g}(q):=\Phi(g,q)$ is a diffeomorphism.

Let $\mathfrak{g}=T_e G$ be the Lie algebra of $G$ and consider the left group action of $G$ on itself, i.e., $\Phi: G \times G \to G$ defined as $\Phi(g,h) = L_{g} (h)= gh $ for all $g,h \in G$. The infinitesimal generator corresponding to $\xi \in \mathfrak{g}$ is $\xi_Q \in \mathfrak{X} (Q)$ which is defined as $\xi_Q(q) = d/dt|_{t=0} \Phi (\exp(t\xi), q)$, where $\exp:\mathfrak g\to G$ denotes the exponential map. 
 
 The Lie bracket on $\mathfrak{g}$ is denoted by $[\cdot,\cdot]$. For each $\xi \in \mathfrak{g}$, the \textit{adjoint map}, $\ad_\xi : \mathfrak{g} \to \mathfrak{g}$ is defined as $\ad_\xi(\eta)= [\xi, \eta]$ for all $\eta \in \mathfrak{g}$. Similarly, the map $\ad_\xi^{*}:\mathfrak{g}^{*}\to\mathfrak{g}^{*}$ denotes the \textit{co-adjoint operator} and is defined as $\langle\ad_{\xi}^{*}(\mu),\eta\rangle=\langle\mu,\ad_{\xi}(\eta)\rangle$ for all $\eta\in\mathfrak{g}$ and $\mu\in\mathfrak g^*$.

For nonholonomic (possibly degenerate) Lagrangian systems on Lie groups that are invariant by the left action of $G$ on itself, the reduced Lagrangian is denoted by $\ell:\mathfrak{g}\to\mathbb{R}$, and the reduction of the Lagrange--d'Alembert equations \eqref{ldaeq} yields the following equations of motion \cite{yoshimura2007reduction}:
\begin{equation}\label{EPSeq}
\mu=\frac{\delta\ell}{\delta\eta},\qquad\xi=\eta\in\mathfrak{d},\qquad\dot{\mu}-\ad_{\xi}^{*}(\mu)\in\mathfrak{d}^{\circ},
\end{equation}
where $\mathfrak{d}\subset\mathfrak{g}$ is the reduced constraint and $\mathfrak{d}^{\circ}\subset\mathfrak{g}^{*}$ denotes its annihilator. These equations are the implicit analog of the Euler--Poincar\'e--Suslov equations \cite{bloch2003nonholonomic}. For this reason, they are called the implicit Euler--Poincar\'e--Suslov equations for nonholonomic mechanics \cite{GaYo2015,yoshimura2007reduction}.

On the other hand, roughly speaking an \textit{Ehresmann connection} specifies how a quantity
associated with a manifold changes as we move from one point to
another; that is to say, it ``connects'' neighboring spaces. In terms of fiber
bundles, a connection tells us how movement in the total space
induces change along the fibers. Recall that a bundle is a triple
$(E,\pi,M)$, where $\pi: E\to M$ is a surjective submersion. The manifolds $M$ and $E$ are knwon as the base space and the total space, respectively,
and the map $\pi$ is known as the projection of the bundle. For each $x\in
M$, the manifold $\pi^{-1}(\{x\})\subset E$ is the fiber of the bundle
over $x\in M$. Given a free and proper (left) action $\Phi:G\times Q\to Q$, the quotient projection defines a \textit{principal bundle} $\pi_{Q,\Sigma}:Q\to \Sigma= Q/G$, where
$\Sigma$ is endowed with the unique manifold structure making $\pi$ a submersion (see, for example, \cite{kobayashi1996foundations}). The manifold $\Sigma$ is called the \textit{shape space}. 

For unconstrained (possibly degenerate) Lagrangian systems on a smooth manifold $Q$ that are invariant by the left action of a Lie grup $G$ on $Q$, a principal connection allos for identifying the quotient $(TQ)/G\simeq T\Sigma\oplus\tilde{\mathfrak g}$. Therefore, the reduced
Lagrangian reads $\ell :T\Sigma \oplus  \tilde{\mathfrak{g}}\to\mathbb{R}$ and the implicit Lagrange--Poincar\'e equations are given by \cite{YoMa2009}
\begin{equation*}
\begin{array}{ll}
\displaystyle\frac{\nabla^{\Sigma*}y}{dt}=\frac{\delta\ell}{\delta\sigma}-\overline\rho\cdot\left(i_{\dot\sigma}\tilde F^\omega\right),\quad & \displaystyle\dot{\sigma}=u,\vspace{1mm}\\
\displaystyle y=\frac{\delta\ell}{\delta u}, & \displaystyle\overline{\xi}=\overline{\eta},\vspace{1mm}\\
\displaystyle\frac{\nabla^{\omega*}\overline{\rho}}{dt}=\ad^{*}_{\tilde{\xi}}\left(\overline{\rho}\right), & \displaystyle\overline{\rho}=\frac{\delta\ell}{\delta\overline{\eta}}.
\end{array}
\end{equation*}

Note that the problem of implicit nonholonomic Lagrange--Poincar\'e reduction has not been considered in \cite{YoMa2009} nor in \cite{GaYo2015}.

\section{Nonholonomic implicit Lagrangian mechanics with collisions}\label{sec3} 

Let $Q$ be a smooth manifold with boundary, denoted by $\partial Q$, $L:TQ\to\mathbb R$ be a (possibly degenerate) Lagrangian, and $\Delta_Q\subset TQ$ be a (possibly nonholonomic) constraint distribution.
According to Section \ref{sec:main}, the annihilator of $\Delta_Q$ is denoted by $\Delta_Q^\circ\subset T^*Q$.

\subsection{Configuration space and phase space}

 Given $[\tau_0,\tau_1]\subset\mathbb R$ and $\tilde\tau\in [\tau_0,\tau_1]$, the \emph{path space with a unique collision (at $\tau=\tilde\tau$)} is defined as $\Omega(Q,\tilde\tau)=\mathcal T\times\mathcal Q(\tilde\tau)$, where
\begin{equation*}
\mathcal T=\left\{\alpha_T\in C^\infty([\tau_0,\tau_1])\mid \alpha_T'(\tau)>0,~\tau\in [\tau_0,\tau_1]\right\}
\end{equation*}
and
\begin{multline}\label{eq:mathcalQ}
\qquad\qquad\mathcal Q(\tilde\tau)=\big\{\alpha_Q\in C^0([\tau_0,\tau_1],Q)\mid\alpha_Q(\tilde\tau)\in\partial Q,\\
\alpha_Q\text{ is piecewise }C^2\text{ and has only one singularity at }\tilde\tau\big\}.\qquad\qquad
\end{multline}
We only consider one singularity at $\tau=\tilde\tau$ for brevity, but similar results hold for a finite amount of singularities, $\{\tilde\tau_i\mid 1\leq i\leq N\}\subset [\tau_0,\tau_1]$.

\begin{remark}[Zeno behaviour]
Systems with collisions are a particular instance of hybrid systems. For systems with elastic impacts, the \emph{guard} is given by $S=\{v_q\in T_q Q\mid q\in\partial Q,~g(v_q,n_q)>0\}$, where $g$ is a Riemannian metric on $Q$ and $n$ is the outward-pointing, unit, normal vector field on the boundary. Similarly, the \emph{reset map} is given by $R(v_q)=v_q^{_\parallel}-v_q^{\perp}$, where $v_q^\perp=g(v_q,n_q)\,n_q$ and $v_q^{_\parallel}=v_q-v_q^\perp\in T_q\partial Q$. Recall that hybrid systems may experience Zeno behaviour if a trajectory undergoes infinitely many impacts in finite time. In order to avoid this situation, we ask the system to satisfy two conditions (cf. \cite[Remark 2.1]{GoCo2020}):
\begin{enumerate}
    \item $S\cap\overline R(S)=\emptyset$, where $\overline R(S)$ is the closure of $R(S)\subset TQ$. This condition is clearly satisfied in our case. Indeed, for each $v_q\in S$ we have $v_q^\perp\neq 0$ and, thus, $\parallel R(v_q)-v_q\parallel_g=2\parallel v_q^\perp\parallel_g>0$, being $\parallel\cdot\parallel_g$ the norm induced by the metric $g$.
    \item The set of collision times is closed and discrete. This condition, which depends on the topology of the configuration manifold, prevents the existence of an accumulation point and will be assumed in the following.
    \end{enumerate}
Under these assumptions, our development is valid in a neighborhood of each collision.
\end{remark}

\begin{lemma}{\cite[Corollary 2.3]{FeMaOrWe2003}}
$\Omega(Q,\tilde\tau)=\mathcal T\times\mathcal Q(\tilde\tau)$ is a smooth manifold.
\end{lemma}

\begin{remark}
Given $\alpha_T\in\mathcal T$, we denote $[t_0,t_1]=\alpha_T([\tau_0,\tau_1])$ and, in order to distinguish between $\tau$-derivatives and $t$-derivatives, we use different symbols; namely, $\alpha_T'=d\alpha_T/d\tau$ and $\dot\alpha_T^{-1}=d\alpha_T^{-1}/dt$, where $\alpha_T^{-1}:[t_0,t_1]\to[\tau_0,\tau_1]$ is the inverse of $\alpha_T$. Analogously, we denote $\tilde t=\alpha_T(\tilde\tau)$.
\end{remark}

The tangent space of $\mathcal Q(\tilde\tau)$ at $\alpha_Q\in\mathcal Q(\tilde\tau)$ is given by
\begin{multline}\label{eq:TmathcalQ}
\qquad\qquad T_{\alpha_Q}\mathcal Q(\tilde\tau)=\Big\{\nu_{\alpha_Q}\in C^0([\tau_0,\tau_1],TQ)\mid\alpha_Q=\pi_{TQ}\circ\nu_{\alpha_Q},~\nu_{\alpha_Q}(\tilde\tau)\in T_{\alpha_Q(\tilde\tau)}\partial Q,\\
\nu_{\alpha_Q}\text{ is piecewise }C^2\text{ and has only one singularity at $\tilde\tau$}\Big\},\qquad\qquad
\end{multline}
where $\pi_{TQ}:TQ\to Q$ is the natural projection. In order to incorporate the constraint distribution, we define the following subspace at each $\alpha_Q\in\mathcal Q(\tilde\tau)$,
\begin{equation*}
\Delta_{\mathcal Q(\tilde\tau)}(\alpha_Q)=\left\{\nu_{\alpha_Q}\in T_{\alpha_Q}\mathcal Q(\tilde\tau)\mid\nu_{\alpha_Q}:[\tau_0,\tau_1]\to\Delta_Q\right\}.
\end{equation*}
As usual, we denote $T\mathcal Q(\tilde\tau)=\bigsqcup_{\alpha_Q\in\mathcal Q(\tilde\tau)}T_{\alpha_Q}\mathcal Q(\tilde\tau)$ and $\Delta_{\mathcal Q(\tilde\tau)}=\bigsqcup_{\alpha_Q\in\mathcal Q(\tilde\tau)}\Delta_{\mathcal Q(\tilde\tau)}(\alpha_Q)$.

Let $T_{\alpha_Q}'\mathcal Q(\tilde\tau)=\{\phi_{\alpha_Q}:T_{\alpha_Q}\mathcal Q(\tilde\tau)\to\mathbb R\mid\phi_{\alpha_Q}\text{ is linear and continuous}\}$ be the topological dual of $T_{\alpha_Q}\mathcal Q(\tilde\tau)$. Since $\mathcal Q(\tilde\tau)$ is an infinite dimensional manifold, its topological cotangent bundle is too large to formulate mechanics. For that reason, we will restrict ourselves to the vector subbundle where the Legendre transform of the Lagrangian lie, i.e., we consider a vector subbundle $T^\star\mathcal Q(\tilde\tau)\subset T'\mathcal Q(\tilde\tau)$ such that $\mathbb FL\circ\nu_{\alpha_Q}\in T^\star\mathcal Q(\tilde\tau))$ for each $\nu_{\alpha_Q}\in T_{\alpha_Q}\mathcal Q(\tilde\tau)$, where $\mathbb FL:TQ\to T^*Q$ is the Legendre transform of $L$.

\begin{lemma}
For each $\alpha_Q\in\mathcal Q(\tilde\tau)$, the vector space
\begin{multline}\label{eq:T*mathcalQ}
\qquad\qquad T_{\alpha_Q}^\star\mathcal Q(\tilde\tau)=\big\{\pi_{\alpha_Q}\in C^0([\tau_0,\tau_1],T^*Q)\mid\alpha_Q=\pi_{T^*Q}\circ\pi_{\alpha_Q},~\pi_{\alpha_Q}(\tilde\tau)\in T_{\alpha_Q(\tau)}^*\partial Q,\\
\pi_{\alpha_Q}\text{ is piecewise }C^2\text{ and has only one singularity at $\tilde\tau$}\big\},\qquad\qquad
\end{multline}
where $\pi_{T^*Q}:T^*Q\to Q$ is the natural projection, is a vector subspace of the topological dual of $T_{\alpha_Q}\mathcal Q(\tilde\tau)$ by means of the following $L^2$-dual pairing:
\begin{equation*}
\langle\pi_{\alpha_Q},\nu_{\alpha_Q}\rangle=\int_{\tau_0}^{\tau_1}\pi_{\alpha_Q}(\tau)\cdot\nu_{\alpha_Q}(\tau)\,d\tau,
\end{equation*}
where $\cdot$ represents the pairing between $T^*Q$ and $TQ$. Furthermore, this pairing is nondegenerate.
\end{lemma}

Note that, in general, $\left\{\mathbb FL\circ\nu_{\alpha_Q}\in T_{\alpha_Q}^\star\mathcal Q(\tilde\tau)\mid\nu_{\alpha_Q}\in T_{\alpha_Q}\mathcal Q(\tilde\tau)\right\}\subsetneq T_{\alpha_Q}^\star\mathcal Q(\tilde\tau)$, as the Lagrangian is possibly degenerate. As a straightforward consequence of the previous lemma, the vector bundle
\begin{equation*}
T^\star\mathcal Q(\tilde\tau)=\bigsqcup_{\alpha_Q\in\mathcal Q(\tilde\tau)}T_{\alpha_Q}^\star\mathcal Q(\tilde\tau)\to\mathcal Q(\tilde\tau),\quad\pi_{\alpha_Q}\mapsto\alpha_Q,
\end{equation*}
is a vector subbundle of the topological cotangent bundle of $\mathcal Q(\tilde\tau)$. 

In the same vein, for each $\nu_{\alpha_Q}\in T_{\alpha_Q}\mathcal Q(\tilde\tau)$ and $\pi_{\alpha_Q}\in T_{\alpha_Q}^\star\mathcal Q(\tilde\tau)$, the iterated bundles are given by
\begin{multline*}
T_{\nu_{\alpha_Q}}(T\mathcal Q(\tilde\tau))=\Big\{\delta\nu_{\alpha_Q}\in C^0([\tau_0,\tau_1],T(TQ))\mid\nu_{\alpha_Q}=\pi_{T(TQ)}\circ\delta\nu_{\alpha_Q},\\~\delta\nu_{\alpha_Q}(\tilde\tau)\in T_{\nu_{\alpha_Q}(\tilde\tau)}(T\partial Q),\\
\delta\nu_{\alpha_Q}\text{ is piecewise $C^2$ and has only one singularity at }\tilde\tau\Big\},
\end{multline*}
where $\pi_{T(TQ)}:T(TQ)\to TQ$ is the natural projection, and
\begin{multline*}
T_{\pi_{\alpha_Q}}(T^\star\mathcal Q(\tilde\tau))=\Big\{\delta\pi_{\alpha_Q}\in C^0([\tau_0,\tau_1],T(T^*Q))\mid\pi_{\alpha_Q}=\pi_{T(T^*Q)}\circ\delta\pi_{\alpha_Q},\\~\delta\pi_{\alpha_Q}(\tilde\tau)\in T_{\pi_{\alpha_Q}(\tilde\tau)}(T^*\partial Q),\\
\delta\pi_{\alpha_Q}\text{ is piecewise $C^2$ and has only one singularity at }\tilde\tau\Big\},
\end{multline*}
where $\pi_{T(T^*Q)}:T(T^*Q)\to T^*Q$ is the natural projection. In particular, we consider the constrained iterated bundle,
\begin{equation}\label{eq:DeltaTQ}
\Delta_{T\mathcal Q(\tilde\tau)}(\nu_{\alpha_Q})=\left\{\delta\nu_{\alpha_Q}\in T_{\nu_{\alpha_Q}}(T\mathcal Q(\tilde\tau))\mid d\pi_{TQ}\circ\delta\nu_{\alpha_Q}\in C^0([\tau_0,\tau_1],\Delta_Q)\right\}.
\end{equation}

\subsection{Nonholonomic implicit Euler--Lagrange equations with collisions}

Given a path $\alpha=(\alpha_T,\alpha_Q)\in\Omega(Q,\tilde\tau)$, the \emph{associated curve} is defined as
\begin{equation}\label{eq:associatedcurve}
q_\alpha:[t_0,t_1]\to Q,\quad t\mapsto q_\alpha(t)=\left(\alpha_Q\circ\alpha_T^{-1}\right)(t).
\end{equation}
Similarly, given $\nu_{\alpha_Q}\in T_{\alpha_Q}\mathcal Q(\tilde\tau)$ and $\pi_{\alpha_Q}\in T_{\alpha_Q}'\mathcal Q(\tilde\tau)$, we set
\begin{equation*}
\begin{array}{ll}
v_\alpha: [t_0,t_1]\to TQ, \quad & t\mapsto v_\alpha(t)=\left(\nu_{\alpha_Q}\circ\alpha_T^{-1}\right)(t),\\
p_\alpha: [t_0,t_1]\to T^*Q, & t\mapsto p_\alpha(t)=\left(\pi_{\alpha_Q}\circ\alpha_T^{-1}\right)(t).
\end{array}
\end{equation*}
It is clear that $\pi_{TQ}\circ v_\alpha=\pi_{T^*Q}\circ p_\alpha=q_\alpha$.

By regarding $\Omega(Q,\tilde\tau)$ as a trivial vector bundle over $\mathcal Q(\tilde\tau)$ with the projection onto the second factor, the \emph{Lagrange--d'Alembert--Pontryagin action functional},
\begin{equation*}
\mathbb S:\Omega(Q,\tilde\tau)\times_{\mathcal Q(\tilde\tau)}\left(T\mathcal Q(\tilde\tau)\oplus T^\star\mathcal Q(\tilde\tau)\right)\to\mathbb R,
\end{equation*}
where $\times_{\mathcal Q(\tilde\tau)}$ denotes the fibered product over $\mathcal Q(\tilde\tau)$, is defined as
\begin{align*}
\mathbb S\left(\alpha,\nu_{\alpha_Q},\pi_{\alpha_Q}\right) & =\int_{t_0}^{t_1}\left(L(v_\alpha(t))+p_\alpha(t)\cdot\left(\dot q_\alpha(t)-v_\alpha(t)\right)\right)dt\\
& =\int_{\tau_0}^{\tau_1}\left(L\left(\nu_{\alpha_Q}(\tau)\right)+\pi_{\alpha_Q}(\tau)\cdot\left(\frac{\alpha_Q'(\tau)}{\alpha_T'(\tau)}-\nu_{\alpha_Q}(\tau)\right)\right)\alpha_T'(\tau)\,d\tau.
\end{align*}
The equality between the first and the second expressions can be easily checked by considering the change of variable $t=\alpha_T(\tau)$. By recalling that the \emph{energy} of the system is given by
\begin{equation}\label{eq:energy}
E:TQ\oplus T^*Q\to\mathbb R,\qquad(v_q,p_q)\mapsto E(v_q,p_q)=p_q\cdot v_q-L(v_q),
\end{equation}
the action functional may be rewritten as
\begin{align*}
\mathbb S\left(\alpha,\nu_{\alpha_Q},\pi_{\alpha_Q}\right) & =\int_{t_0}^{t_1}\left(p_\alpha(t)\cdot\dot q_\alpha(t)-E(v_\alpha(t),p_\alpha(t)\right)dt\\
& =\int_{\tau_0}^{\tau_1}\left(\pi_{\alpha_Q}(\tau)\cdot\frac{\alpha_Q'(\tau)}{\alpha_T'(\tau)}-E\left(\nu_{\alpha_Q}(\tau),\pi_{\alpha_Q}(\tau)\right)\right)\alpha_T'(\tau)\,d\tau.
\end{align*}

\begin{definition}[Hamilton--d'Alembert--Pontryagin  principle]\label{def:HdAPprinciple}
A path
\begin{equation*}
\texttt{c}=((\alpha_T,\alpha_Q),\nu_{\alpha_Q},\pi_{\alpha_Q})\in\Omega(Q,\tilde\tau)\times_{\mathcal Q(\tilde\tau)}\left(\Delta_{\mathcal Q(\tilde\tau)}\oplus T^\star\mathcal Q(\tilde\tau)\right)
\end{equation*}
is \emph{stationary} (or \emph{critical}) for the action functional $\mathbb S$ if it satisfies
\begin{equation*}
d\mathbb S(\texttt{c})(\delta\texttt{c})=0
\end{equation*}
for every variation $\delta\texttt{c}=\left((\delta\alpha_T,\delta\alpha_Q),\delta\nu_{\alpha_Q},\delta\pi_{\alpha_Q}\right)\in T_\alpha\Omega(Q,\tilde\tau)\times\Delta_{T\mathcal Q(\tilde\tau)}(\nu_{\alpha_Q})\times T_{\pi_{\alpha_Q}}(T^\star\mathcal Q(\tilde\tau))$ such that $\delta\alpha_T(\tau_0)=\delta\alpha_T(\tau_1)=0$, $\delta\alpha_Q(\tau_0)=\delta\alpha_Q(\tau_1)=0$ and
\begin{equation}\label{eq:projectionvariation}
d\pi_{TQ}\circ\delta\nu_{\alpha_Q}=d\pi_{T^*Q}\circ\delta\pi_{\alpha_Q}=\delta\alpha_Q.
\end{equation}
\end{definition}

In order to give an intrinsic expression for the dynamical equations, we fix a linear connection on the tangent bundle of $Q$, $\nabla^Q:\mathfrak X(Q)\to\Omega^1(Q,TQ)$, as well as its dual, which is a linear connection on the cotangent bundle of $Q$, $\nabla^{Q*}:\Omega^1(Q)\to\Omega^1(Q,T^*Q)$. These connections enable us to compute the vertical part of the variations. For instance, given $\alpha_Q\in\mathcal Q(\tilde\tau)$, $\nu_{\alpha_Q}\in T_{\alpha_Q}\mathcal Q(\tilde\tau)$ and $\delta\nu_{\alpha_Q}\in T_{\nu_{\alpha_Q}}(T\mathcal Q(\tilde\tau))$, we write $\delta^Q\nu_{\alpha_Q}=\left(\delta\nu_{\alpha_Q}\right)^v\in T_{\alpha_Q}\mathcal Q(\tilde\tau)$, where the superscript $v$ denotes the vertical part given by $\nabla^Q$, and analogous for the dual connection: $\delta^{Q*}$. Additionally, the covariant derivatives induced by these connections are denoted by $\nabla^Q/d\tau$ and $\nabla^{Q*}/d\tau$, respectively.

The partial functional derivatives of the Lagrangian are denoted by
\begin{equation*}
\begin{array}{ll}
\displaystyle\frac{\delta L}{\delta v_q}:TQ\to T^*Q,\quad & \displaystyle\frac{\delta L}{\delta v_q}(v_q)\cdot w_q=\left.\frac{d}{ds}\right|_{s=0}L(v_q+s\,w_q),\vspace{2mm}\\
\displaystyle\frac{\delta L}{\delta q}:TQ\to T^*Q,\quad & \displaystyle\frac{\delta\ell}{\delta q}(v_q)\cdot w_q=\left.\frac{d}{ds}\right|_{s=0}\left(L\circ\gamma_{v_q}^h\right)(s),
\end{array}
\end{equation*}
for each $v_q,w_q\in T_q Q$, $q\in Q$, where $\gamma:(-\epsilon,\epsilon)\to Q$ is a curve such that $d/ds|_{s=0}\gamma=w_q$ and $\gamma_{v_q}^h:(-\epsilon,\epsilon)\to TQ$ is the horizontal lift at $v_q$ given by $\nabla^Q$. Observe that the first one is a vertical (fiber) derivative, whereas the second one is a horizontal derivative and depends on the choice of the connection.

\begin{theorem}[Nonholonomic implicit Euler--Lagrange equations with collisions]\label{theorem:implicitELequations}
A path
\begin{equation*}
((\alpha_T,\alpha_Q),\nu_{\alpha_Q},\pi_{\alpha_Q})\in\Omega(Q,\tilde\tau)\times_{\mathcal Q(\tilde\tau)}\left(\Delta_{\mathcal Q(\tilde\tau)}\oplus T^\star\mathcal Q(\tilde\tau)\right)
\end{equation*} 
is critical for the action functional $\mathbb S$ if and only if it satisfies the \emph{implicit Euler--Lagrange equations}, or \emph{Lagrange--d'Alembert--Pontryagin equations}:
\begin{equation*}
\left\{\begin{array}{ll}
\displaystyle\frac{\nabla^{Q*}\pi_{\alpha_Q}}{d\tau}-\alpha_T'\frac{\delta L}{\delta q}(\nu_{\alpha_Q})\in\Delta_Q^\circ(\alpha_Q),\qquad & \displaystyle E'(\nu_{\alpha_Q},\pi_{\alpha_Q})=0,\\
\displaystyle\pi_{\alpha_Q}=\frac{\delta L}{\delta v_q}(\nu_{\alpha_Q}), & \displaystyle\nu_{\alpha_Q}=\frac{\alpha_Q'}{\alpha_T'}\in\Delta_Q(\alpha_Q),
\end{array}\right.
\end{equation*}
on $[\tau_0,\tilde\tau)\cup(\tilde\tau,\tau_1]$, together with the conditions for the \emph{elastic impact},
\begin{equation*}
\left\{\begin{array}{l}
\displaystyle \pi_{\alpha_Q}^+-\pi_{\alpha_Q}^-\in\left(T_{\alpha_Q(\tilde\tau)}\partial Q\cap\Delta_Q(\alpha_Q(\tilde\tau))\right)^\circ=\left(T_{\alpha_Q(\tilde\tau)}\partial Q\right)^\circ+\Delta_Q^\circ(\alpha_Q(\tilde\tau)),\vspace{2mm}\\
\displaystyle E(\nu_{\alpha_Q}^+,\pi_{\alpha_Q}^+)=E(\nu_{\alpha_Q}^-,\pi_{\alpha_Q}^-),
\end{array}\right.
\end{equation*}
where the annihilators are with respect to $TQ$ and we denote $\pi_{\alpha_Q}(\tilde\tau^+)=\pi_{\alpha_Q}^+$, etc.
\end{theorem}

\begin{proof}
Firstly, \eqref{eq:projectionvariation} ensures that the horizontal part of the variations is $\delta\alpha_Q\in\Delta_{\mathcal Q(\tilde\tau)}(\alpha_Q)$. As a result, the variation of the action functional reads
\begin{align*}
d\mathbb S(\alpha,\nu_{\alpha_Q},\pi_{\alpha_Q})\left(\delta\alpha,\delta\nu_{\alpha_Q},\delta\pi_{\alpha_Q}\right) & =\int_{\tau_0}^{\tau_1}\Bigg(\frac{\partial L}{\partial q}\cdot\delta\alpha_Q+\frac{\partial L}{\partial v_q}\cdot\delta^Q\nu_{\alpha_Q}+\delta^{Q*}\pi_{\alpha_Q}\cdot\left(\frac{\alpha_Q'}{\alpha_T'}-\nu_{\alpha_Q}\right)\\
& \hspace{15mm}+\pi_{\alpha_Q}\cdot\left(\frac{\delta^Q\alpha_Q'}{\alpha_T'}-\frac{\alpha_Q'\delta\alpha_T'}{(\alpha_T')^2}-\delta^Q\nu_{\alpha_Q}\right)\Bigg)\,\alpha_T'\,d\tau\\
& +\int_{\tau_0}^{\tau_1}\left(L+\pi_{\alpha_Q}\cdot\left(\frac{\alpha_Q'}{\alpha_T'}-\nu_{\alpha_Q}\right)\right)\delta\alpha_T'\,d\tau,
\end{align*}
where the Lagrangian, as well as its partial derivatives, are evaluated at $\nu_{\alpha_Q}(\tau)$. Note that
\begin{equation*}
\delta^Q\alpha_Q'=\frac{\nabla^Q\alpha_Q}{d\tau},\qquad(\pi_{\alpha_Q}\cdot\delta\alpha_Q)'=\frac{\nabla^{Q*}\pi_{\alpha_Q}}{d\tau}\cdot\delta\alpha_Q+\pi_{\alpha_Q}\cdot\frac{\nabla^Q\delta\alpha_Q}{d\tau}.
\end{equation*}
By using this, splitting the integration domain, $[\tau_0,\tau_1]-\{\tilde\tau\}=[\tau_0,\tilde\tau)\cup(\tilde\tau,\tau_1]$, integrating by parts on each sub-interval and regrouping terms, we may rewrite the previous expression as
\begin{equation*}
d\mathbb S(\alpha,\nu_{\alpha_Q},\pi_{\alpha_Q})\left(\delta\alpha,\delta\nu_{\alpha_Q},\delta\pi_{\alpha_Q}\right)=\mathcal I(\tau_0,\tilde\tau)+\mathcal I(\tilde\tau,\tau_1)+\mathcal B(\tau_0,\tilde\tau^-)+\mathcal B(\tilde\tau^+,\tau_1),
\end{equation*}
where for each $a,b\in\mathbb R$, $a<b$, we set
\begin{align*}
\mathcal I(a,b) & =\int_a^b\Bigg(\left(\alpha_T'\frac{\delta L}{\delta q}-\frac{\nabla^{Q*}\pi_{\alpha_Q}}{d\tau}\right)\cdot\delta\alpha_Q+\alpha_T'\left(\frac{\delta L}{\delta v_q}-\pi_{\alpha_Q}\right)\cdot\delta^Q\nu_{\alpha_Q}\\
& \hspace{15mm}+\delta^{Q*}\pi_{\alpha_Q}\left(\frac{\alpha_Q'}{\alpha_T'}-\nu_{\alpha_Q}\right)\alpha_T'+E'\,\delta\alpha_T\Bigg)\,d\tau\\
\mathcal B(a,b)& =\Big[\pi_{\alpha_Q}(\tau)\cdot\delta\alpha_Q(\tau)-E\,\delta\alpha_T(\tau)\Big]_{\tau=a}^{\tau=b}.
\end{align*}
where the energy is evaluated at $(\nu_{\alpha_Q}(\tau),\pi_{\alpha_Q}(\tau))$. Since the previous expression vanishes for arbitrary variations $\delta\alpha_T\in T_{\alpha_T}\mathcal T$ and $\delta\alpha_Q\in\Delta_{\mathcal Q(\tilde\tau)}(\alpha_Q)$ vanishing at the endpoints, as well as arbitrary variations $\delta^Q\nu_{\alpha_Q}\in T_{\alpha_Q}\mathcal Q(\tilde\tau)$, $\delta^{Q*}\pi_{\alpha_Q}\in T_{\alpha_Q}^\star\mathcal Q(\tilde\tau)$, we obtain the desired equations, together with the impact conditions.
\end{proof}

By using the change of variable $t=\alpha_T(\tau)$, we have $\dot q_\alpha=\alpha_Q'/\alpha_T'$ and $\dot p_\alpha=\pi_{\alpha_Q}'/\alpha_T'$. Then, the implicit Euler--Lagrange equations for a curve $(v_\alpha,p_\alpha):[t_0,t_1]\to TQ\oplus T^*Q$ take the form
\begin{equation}\label{eq:ELeqsimplicit}\left\{\begin{array}{ll}
\displaystyle\frac{\nabla^{Q*} p_\alpha}{dt}-\frac{\delta L}{\delta q}(v_\alpha)\in\Delta_Q^\circ(q_\alpha),\qquad & \displaystyle \dot E(v_\alpha,p_\alpha)=0,\vspace{0.1cm}\\
\displaystyle p_\alpha=\frac{\delta L}{\delta v_\alpha}(v_\alpha), & \displaystyle v_\alpha=\dot q_\alpha\in\Delta_Q(q_\alpha),
\end{array}\right.
\end{equation}
on $\left[t_0,\tilde t\right)\cup\left(\tilde t,t_1\right]$. Similarly, the conditions for the elastic impact read
\begin{align}\label{eq:impactL}
& p_\alpha^+-p_\alpha^-\in\left(T_{q_\alpha\left(\tilde t\right)}\partial Q\cap\Delta_Q\left(q_\alpha\left(\tilde t\right)\right)\right)^\circ=\left(T_{q_\alpha\left(\tilde t\right)}\partial Q\right)^\circ+\Delta_Q^\circ\left(q_\alpha\left(\tilde t\right)\right),\\\nonumber
& E\left(v_\alpha^+,p_\alpha^+\right)=E\left(v_\alpha^-,p_\alpha^-\right),\\\label{eq:impactDelta}
& v_\alpha^+=\dot q_\alpha^+\in\Delta_Q\left(q_\alpha\left(\tilde t\right)\right),
\end{align}
where we denote $p_\alpha\left(\tilde t^+\right)=p_\alpha^+$, etc.

For unconstrained systems, i.e., $\Delta_Q=TQ$, the Hamilton--d'Alembert--Pontryagin principle reduces to the \emph{Hamilton--Pontryagin principle}, and the implicit Euler--Lagrange equations of motion read as
\begin{equation*}\left\{\begin{array}{ll}
\displaystyle\frac{\nabla^{Q*}p_\alpha}{dt}=\frac{\delta L}{\delta q}(v_\alpha),\qquad & \displaystyle \dot E(v_\alpha,p_\alpha)=0,\vspace{0.1cm}\\
\displaystyle p_\alpha=\frac{\delta L}{\delta v_q}(v_\alpha), & \displaystyle v_\alpha=\dot q_\alpha.
\end{array}\right.
\end{equation*}

\begin{remark}[Energy balance]
The conservation of the energy, $\dot E(v_\alpha,p_\alpha)=0$, is redundant, as it may be obtained from the remaining equations:
\begin{align*}
\dot E(v_\alpha,p_\alpha) & =\frac{\nabla^{Q*}p_\alpha}{dt}\cdot v_\alpha+p_\alpha\cdot\frac{\nabla^Q v_\alpha}{dt}-\frac{\delta L}{\delta q}\cdot\dot q_\alpha-\frac{\delta L}{\delta v_q}\cdot\frac{\nabla^Q v_\alpha}{dt}\\
& =\left(\frac{\nabla^{Q*}p_\alpha}{dt}-\frac{\delta L}{\delta q}\right)\cdot v_\alpha=0,
\end{align*}
where the last equality from the fact that $v_\alpha\in\Delta_Q(q_\alpha)$ while $\nabla^{Q*}p_\alpha/dt-\delta L/\delta q\in\Delta_Q^\circ(q_\alpha)$. Consequently, this equation may be omitted.
\end{remark}


\subsection{Example: The rolling disk hitting a circular surface}\label{ex:disk}

Let us consider a disk rolling without slipping, as in \cite[\S 7.1]{YoMa2006a} or \cite[\S VI]{simoes2023hamel}, that is confined to move in a solid circle. The configuration space is thus given by
\begin{equation*}
Q=\{(x,y,\theta,\varphi)\in\mathbb R^2\times S^1\times S^1\mid (x+R\cos\varphi)^2+(y+R\sin\varphi)^2\leq 1\},
\end{equation*}
where $(x,y)$ denotes the contact point of the disk with the ground, $\theta$ denotes the angle of rotation and $\varphi$ denotes the heading angle of the disk with respect to the $x$-axis. The Lagrangian $L:TQ\to\mathbb R$ is given by
\begin{equation*}
L(x,y,\theta,\varphi;v_x,v_y,v_\theta,v_\varphi)=\frac{1}{2}m\left(v_x^2+v_y^2\right)+\frac{1}{2}\left(I\,v_\theta^2+J\,v_\varphi^2\right),
\end{equation*}
where $m,I,J\in\mathbb R^+$ are the mass and the moments of inertia of the disk, respectively. For each $(v_q,p_q)=(x,y,\theta,\varphi;v_x,v_y,v_\theta,v_\varphi;p_x,p_y,p_\theta,p_\varphi)\in TQ\oplus T^*Q$, the energy reads
\begin{align*}
E(v_q,p_q)= p_x\,v_x+p_y\,v_y+p_\theta\,v_\theta+p_\varphi\,v_\varphi-\frac{1}{2}m\left(v_x^2+v_y^2\right)-\frac{1}{2}\left(I\,v_\theta^2+Jv_\varphi^2\right).
\end{align*}
In addition, the non-slipping condition reads $v_x=R\, v_\theta\cos\varphi,$ $v_y=R\,v_\theta\sin\varphi$, where $R\in\mathbb R^+$ is the radius of the disk, thus yielding the following non-holonomic constraint:
\begin{equation*}
\Delta_Q=\operatorname{span}\{\partial_\theta+R\cos\varphi\,\partial_x+R\sin\varphi\,\partial_y,\partial_\varphi\},\quad\Delta_Q^\circ=\operatorname{span}\{dx-R\cos\varphi\,d\theta,dy-R\sin\varphi\,d\theta\}.
\end{equation*}
On the other hand, the boundary of the configuration manifold is given by
\begin{equation*}
\partial Q=\{(x,y,\theta,\varphi)\in\mathbb R^2\times S^1\times S^1\mid (x+R\cos\varphi)^2+(y+R\sin\varphi)^2=1\}.
\end{equation*}
Hence, the tangent bundle of the boundary reads
\begin{equation*}
T\partial Q=\operatorname{span}\{R(x\sin\varphi-y\cos\varphi)\partial_x+(x+R\cos\varphi)\partial_\varphi,R(x\sin\varphi-y\cos\varphi)\partial_y+(y+R\sin\varphi)\partial_\varphi,\partial_\theta\}.
\end{equation*}
and its annihilator is given by
\begin{equation*}
(T\partial Q)^\circ=\operatorname{span}\{(x+R\cos\varphi)dx+(y+R\sin\varphi)dy+R(-x\sin\varphi+y\cos\varphi)d\varphi\}.
\end{equation*}

Lastly, since $S^1$ is a Lie group, its tangent bundle is trivial and, thus, the canonical flat connection on the tangent bundle of $Q$ may be chosen, i.e., $\nabla^Q/dt=d/dt$. By gathering all, the implicit Euler--Lagrange equations \eqref{eq:ELeqsimplicit} for a curve
\begin{equation*}
(x,y,\theta,\varphi;v_x,v_y,v_\theta,v_\varphi;p_x,p_y,p_\theta,p_\varphi):[t_0,t_1]\to TQ\oplus T^*Q    
\end{equation*}
read
\begin{equation*}
\left\{\begin{array}{ll}
\dot p_x=\mu_1,\qquad & \dot p_y=\mu_2,\\
\dot p_\theta=-\mu_1\,R\cos\varphi-\mu_2\, R\sin\varphi,\qquad & \dot p_\varphi=0,\\
p_x=m\,v_x,\qquad & p_y=m\,v_y,\\
p_\theta=I\,v_\theta,\qquad & p_\varphi=J\,v_\varphi,\\
v_x=\dot x=\mu_3\,R\cos\varphi,\qquad & v_y=\dot y=\mu_3\,R\sin\varphi,\\
v_\theta=\dot\theta=\mu_3,\qquad & v_\varphi=\dot\varphi=\mu_4,\\
\end{array}\right.
\end{equation*}
on $[t_0,t_1]-\left\{\tilde t\right\}$, where $\mu_1,\mu_2,\mu_3,\mu_4\in\mathbb R$ are the Lagrange multipliers. The impact condition at $t=\tilde t$ given in \eqref{eq:impactL} reads
\begin{equation*}
\left\{\begin{array}{ll}
p_x^+-p_x^-=\lambda_0(x+R\cos\varphi)+\lambda_1,\qquad & p_\theta^+-p_\theta^-=-\lambda_1\,R\cos\varphi-\lambda_2\,R\sin\varphi,\\
p_y^+-p_y^-=\lambda_0(y+R\sin\varphi)+\lambda_2,\qquad & p_\varphi^+-p_\varphi^-=\lambda_0\,R(-x\sin\varphi+y\cos\varphi),\\
\end{array}\right.
\end{equation*}
where we denote $p_x^+=p_x\left(\tilde t^+\right)$, etc., and $\lambda_0,\lambda_1,\lambda_2\in\mathbb R$ are the Lagrange multipliers. Similarly, the condition \eqref{eq:impactDelta} reads
\begin{equation*}
v_x^+=\lambda_3\,R\cos\varphi,\qquad v_y^+=\lambda_3\,R\sin\varphi,\qquad v_\theta^+=\lambda_3,\qquad v_\varphi^+=\lambda_4,
\end{equation*}
where $v_x^+=v_x\left(\tilde t^+\right)$, etc., and $\lambda_3,\lambda_4\in\mathbb R$ are the Lagrange multipliers.

\subsection{Example: The spherical pendulum hitting a cylindrical surface}\label{ex:pendulum}

Let us consider a spherical pendulum (cf. \cite[\S 5.1]{EyCoBl2021}) hitting a cylindrical surface. The configuration space of the system and its boundary are given by
\begin{equation*}
Q=\{(\theta,\varphi)\in S^2\mid L\sin\theta\leq 1\},\qquad\partial Q=\{(\theta,\varphi)\in S^2\mid L\sin\theta=1\},
\end{equation*}
where $0<L<1$ is the length of the pendulum, and the Lagrangian reads
\begin{equation}\label{eq:lagrangianpendulum}
L(\theta,\varphi;v_\theta,v_\varphi)=\frac{1}{2}mL^2(v_\theta^2+v_\varphi^2\sin^2\theta)-mgL\cos\theta,
\end{equation}
where $m,g\in\mathbb R^+$ are the mass of the pendulum and the gravitational acceleration, respectively. Let us suppose that the polar and azimuthal velocities are proportionally related by a function depending only on the polar angle, i.e., $v_\varphi=f(\theta) v_\theta$ with $f(\theta)>0$ for each $\theta\in\mathbb R$. This results in the following non-holonomic constraint:
\begin{equation}\label{eq:constraintpendulum}
\Delta_Q=\operatorname{span}\{\partial_\theta+f(\theta)\partial_\varphi\},\qquad\Delta_Q^\circ=\operatorname{span}\{f(\theta)d\theta-d\varphi\}.
\end{equation}
In addition, the tangent bundle of the reduced space and the corresponding annihilator read
\begin{equation*}
T\partial Q=\operatorname{span}\{\partial_\varphi\},\qquad (T\partial Q)^\circ=\operatorname{span}\{d\theta\}.
\end{equation*}
By considering the flat connection in these coordinates, the implicit Euler--Lagrange equations \eqref{eq:ELeqsimplicit} for a curve
\begin{equation*}
(\theta,\varphi;v_\theta,v_\varphi;p_\theta,p_\varphi):[t_0,t_1]\to TQ\oplus T^*Q,
\end{equation*}
read
\begin{equation*}
\left\{\begin{array}{ll}
\displaystyle \dot p_\theta-mL\sin\theta(Lv_\varphi^2\cos\theta+g)=\mu_0 f(\theta),\qquad & \displaystyle\dot p_\varphi=-\mu_0\vspace{0.1cm}\\
\displaystyle p_\theta=mL^2v_\theta, & \displaystyle p_\varphi=mL^2v_\varphi\sin^2\theta\vspace{0.1cm}\\
v_\theta=\dot\theta=\mu_1,\qquad & v_\varphi=\dot\varphi=f(\theta)\mu_1.
\end{array}\right.
\end{equation*}
on $[t_0,t_1]-\left\{\tilde t\right\}$, where $\mu_0,\mu_1\in\mathbb R$ are the Lagrange multipliers. The impact condition at $t=\tilde t$ given in \eqref{eq:impactL} reads
\begin{equation*}
p_\theta^+-p_\theta^-=\lambda_0+f(\theta)\lambda_1,\qquad p_\varphi^+-p_\varphi^-=-\lambda_1.
\end{equation*}
where we denote $p_\theta^+=p_\theta\left(\tilde t^+\right)$, etc., and $\lambda_0,\lambda_1\in\mathbb R$ are the Lagrange multipliers. Similarly, the condition \eqref{eq:impactDelta} reads
\begin{equation*}
v_\theta^+=\dot\theta^+=\lambda_3,\qquad v_\varphi^+=\dot\varphi^+=f(\theta)\lambda_3.
\end{equation*}
where $v_\theta^+=v_\theta\left(\tilde t^+\right)$, etc., and $\lambda_3\in\mathbb R$ is the Lagrange multiplier.

\section{Nonholonomic implicit Lagrange--Poincaré reduction with collisions}\label{sec4}

Let $G$ be a Lie group and $\Phi:G\times Q\to Q$ be a free and proper left action. The corresponding quotient, which is a smooth manifold, is denoted by $\Sigma=Q/G$. Furthermore, the quotient projection, $\pi_{Q,\Sigma}:Q\to \Sigma$, is a principal $G$-bundle. As usual, we denote $\Phi(g,q)=\Phi_g(q)=g\cdot q$ for each $(g,q)\in G\times Q$. Henceforth, the equivalence classes induced by the $G$-action on the different spaces are denoted by square brackets $[\cdot]_G$, e.g., $[q]_G\in\Sigma$ for each $q\in Q$ and $[v_q]_G\in (TQ)/G$ for each $v_q\in TQ$.

\begin{remark}
Since $\Phi_g:Q\to Q$ is a diffeomorphism for each $g\in G$, it leaves the boundary invariant, i.e., $\Phi_g(\partial Q)=\partial Q$. Therefore, $\Phi$ induces a left $G$-action on the boundary, $\Phi_\partial=\Phi|_{G\times\partial Q}:G\times\partial Q\to\partial Q$. This action is again free and proper and, thus, the quotient is a smooth manifold, $(\partial Q)/G$. In fact, this quotient is nothing but the boundary of $\Sigma=Q/G$, i.e., $\partial\Sigma=(\partial Q)/G$. Observe that, in particular, we have that $\dim G<\dim Q$.
\end{remark}

The moment map induced by this action, $\mathbf J:T^*Q\to\mathfrak g^*$, is denoted by the condition
\begin{equation*}
\mathbf J(p_q)\cdot\xi=p_q\cdot\xi_q^*,\qquad q\in Q,~p_q\in T^*Q,~\xi\in\mathfrak g,
\end{equation*}
where $\xi^*\in\mathfrak X(Q)$ denotes the infinitesimal generator (or fundamental vector field) of $\xi$, i.e., $\xi_q^*=d/dt|_{t=0}\exp(t\xi)\cdot q$, being $\exp:\mathfrak g\to G$ the exponential map. In addition, the $G$-action may be lifted to the tangent and the cotangent bundles of $Q$,
\begin{equation}\label{eq:GTQ}
\left\{\begin{array}{ll}
G\times TQ\to TQ,\quad & (g,v_q)\mapsto\left(d\Phi_g\right)_q(v_q),\\
G\times T^*Q\to T^*Q,\quad & (g,p_q)\mapsto\left(d\Phi_{g^{-1}}\right)_{g\cdot q}^*(p_q),
\end{array}\right.
\end{equation}
where $\left(d\Phi_{g^{-1}}\right)_{g\cdot q}^*:T_q^*Q\to T_{g\cdot q}^*Q$ is the adjoint map of $\left(d\Phi_{g^{-1}}\right)_{g\cdot q}:T_{g\cdot q} Q\to T_q Q$. In turn, these actions may be lifted to the iterated bundles, yielding $G\times T(TQ)\to T(TQ)$ and $G\times T(T^*Q)\to T(T^*Q)$.

\subsection{Reduced configuration space and reduced constraint distribution}

 Let $\omega\in\Omega^1(Q,\mathfrak g)$ be a principal connection on $\pi_{Q,\Sigma}:Q\to \Sigma$, where $\mathfrak g$ is the Lie algebra of $G$. Let us denote by $\tilde Q=\pi_{Q,\Sigma}^*\left(T\Sigma \right)$ the pullback of the tangent bundle of $\Sigma$ by the quotient projection. In the same fashion, we denote $\tilde Q^*=\pi_{Q,\Sigma}^*\left(T^*\Sigma \right)$. The principal connection thus induces the following trivializations of the tangent and the cotangent bundles of $Q$, 
\begin{equation}\label{eq:trivialTQ}
\left\{\begin{array}{ll}
TQ\simeq\tilde Q\times\mathfrak g,\quad & v_q\mapsto\left((q,(d\pi_{Q,\Sigma})_q(v_q)),\omega_q(v_q)\right),\\
T^*Q\simeq\tilde Q^*\times\mathfrak g^*,\quad & p_q\mapsto\left((q,\texttt{H}_q^*(p_q)),\mathbf J(p_q)\right),
\end{array}\right.
\end{equation}
where $\texttt{H}_q^*:T_q^*Q\to T_\sigma^*\Sigma $ is the dual map of the horizontal lift induced by $\omega$, $\texttt{H}_q:T_\sigma\Sigma \to T_q Q$, being $\sigma=\pi_{Q,\Sigma}(q)$. Under these identifications, it is easy to check that the $G$-actions \eqref{eq:GTQ} read as 
\begin{equation}\label{eq:trivialGactions}
\left\{\begin{array}{ll}
g\cdot((q,v_\sigma),\xi)=((g\cdot q,v_\sigma),\Ad_g(\xi)),\qquad & (q,v_\sigma)\in\tilde Q,~\xi\in\mathfrak g,\\
g\cdot((q,p_\sigma),\rho)=((g\cdot q,p_\sigma),\Ad_{g^{-1}}^*(\rho)),\qquad & (q,p_\sigma)\in\tilde Q^*,~\rho\in\mathfrak g^*,
\end{array}\right.
\end{equation}
for each $g\in G$, where $\Ad:G\to\operatorname{Aut}(\mathfrak g)$ and $\Ad^*:G\to\operatorname{Aut}(\mathfrak g^*)$ denote the adjoint and the coadjoint representations of $G$, respectively. As a result, we have the following identifications for the quotient spaces,
\begin{equation}\label{eq:TQG}
\left\{\begin{array}{ll}
(TQ)/G\simeq T\Sigma\oplus\tilde{\mathfrak g},\quad & [v_q]\mapsto\left((d\pi_{Q,\Sigma})_q(v_q),[q,\omega_q(v_q)]_G\right)\\
(T^*Q)/G\simeq T^*\Sigma \oplus\tilde{\mathfrak g}^*,\quad & [p_q]_G\mapsto\left(\texttt{H}_q^*(p_q),[q,\mathbf J(p_q)]_G\right),
\end{array}\right.
\end{equation}
where $\tilde{\mathfrak g}=(Q\times\mathfrak g)/G$ and $\tilde{\mathfrak g}^*=(Q\times\mathfrak g^*)/G$ are the adjoint and coadjoint bundles, respectively.

Given $[\tau_0,\tau_1]\subset\mathbb R$ and $\tilde\tau\in[\tau_0,\tau_1]$, path spaces analogous to \eqref{eq:mathcalQ}, \eqref{eq:TmathcalQ} and \eqref{eq:T*mathcalQ} may be defined by exchanging $Q$ by $\Sigma$. Such spaces are denoted by $\boldsymbol\Sigma(\tilde\tau)$, $T_{\alpha_\Sigma}\boldsymbol\Sigma(\tilde\tau)$ and $T_{\alpha_\Sigma}^\star\boldsymbol\Sigma(\tilde\tau)$, respectively, where $\alpha_\Sigma\in\boldsymbol\Sigma(\tilde\tau)$. Hence, the \emph{reduced path space} is given by $\Omega(\Sigma,\tilde\tau)=\mathcal T\times\boldsymbol\Sigma(\tilde\tau)$. In the same vein, we define
\begin{multline*}
\qquad\qquad\qquad\tilde{\boldsymbol{\mathfrak g}}(\tilde\tau)_{\alpha_\Sigma}=\big\{\xi_{\alpha_\Sigma}\in C^0\left([\tau_0,\tau_1],\tilde{\mathfrak g}\right)\mid\pi_{\tilde{\mathfrak g}}\circ\xi_{\alpha_\Sigma}=\alpha_\Sigma,\\
\xi_{\alpha_\Sigma}\text{ is piecewise }C^2\text{ and has only one singularity at }\tilde\tau\big\},\qquad\qquad\qquad
\end{multline*}
where $\pi_{\tilde{\mathfrak g}}:\tilde{\mathfrak g}\to\Sigma$ is the natural projection, and analogous for $\tilde{\boldsymbol{\mathfrak g}}^\star(\tilde\tau)_{\alpha_\Sigma}$. As a result, we obtain two vector bundles over $\boldsymbol\Sigma(\tilde\tau)$; namely, $\tilde{\boldsymbol{\mathfrak g}}(\tilde\tau)=\bigsqcup_{\alpha_\Sigma\in\boldsymbol\Sigma(\tilde\tau)}\tilde{\boldsymbol{\mathfrak g}}(\tilde\tau)_{\alpha_\Sigma}\to\boldsymbol\Sigma(\tilde\tau)$ and $\tilde{\boldsymbol{\mathfrak g}}^\star(\tilde\tau)=\bigsqcup_{\alpha_\Sigma\in\boldsymbol\Sigma(\tilde\tau)}\tilde{\boldsymbol{\mathfrak g}}^\star(\tilde\tau)_{\alpha_\Sigma}\to\boldsymbol\Sigma(\tilde\tau)$.

Lastly, let $\Delta_Q\subset TQ$ be a $G$-equivariant constraint distribution, i.e.,
\begin{equation*}\label{eq:equivarianceDelta}
\Delta_Q(g\cdot q)=(d\Phi_g)_q(\Delta_Q(q)),\qquad g\in G,~q\in Q.
\end{equation*}
Hence, it may be dropped to a distribution on the reduced bundle, $[\Delta_Q]_G\subset(TQ)/G$. By means of \eqref{eq:TQG}, we may write
\begin{equation}\label{eq:reducedDeltaQ}
[\Delta_Q]_G\simeq\Delta_\Sigma\oplus\Delta_{\tilde{\mathfrak g}}\subset T\Sigma\oplus\tilde{\mathfrak g},    
\end{equation}
where $\Delta_\Sigma$ and $\Delta_{\tilde{\mathfrak g}}$ are the \emph{horizontal} and the \emph{vertical reduced constraint distributions}. This way, constrained paths may be considered:
\begin{align*}
\Delta_{\boldsymbol\Sigma(\tilde\tau)}(\alpha_\Sigma) & =\left\{\nu_{\alpha_\Sigma}\in T_{\alpha_\Sigma}\boldsymbol\Sigma(\tilde\tau)\mid\nu_{\alpha_\Sigma}:[\tau_0,\tau_1]\to\Delta_\Sigma\right\},\\
\Delta_{\tilde{\boldsymbol{\mathfrak g}}(\tilde\tau)}(\alpha_\Sigma) & =\left\{\xi_{\alpha_\Sigma}\in\tilde{\boldsymbol{\mathfrak g}}(\tilde\tau)_{\alpha_\Sigma}\mid\xi_{\alpha_\Sigma}:[\tau_0,\tau_1]\to\Delta_{\tilde{\mathfrak g}}\right\},
\end{align*}
for each $\alpha_\Sigma\in\boldsymbol\Sigma(\tilde\tau)$. Note that we obtain vector subbundles $\Delta_{\boldsymbol\Sigma(\tilde\tau)}\subset T\boldsymbol\Sigma(\tilde\tau)$ and $\Delta_{\tilde{\boldsymbol{\mathfrak g}}(\tilde\tau)}\subset\tilde{\boldsymbol{\mathfrak g}}(\tilde\tau)$. The following lemma is now straightforward.

\begin{lemma}\label{lemma:reducedpath}
In the above conditions, let
\begin{equation*}
\texttt{c}=((\alpha_T,\alpha_Q),\nu_{\alpha_Q},\pi_{\alpha_Q})\in\Omega(Q,\tilde\tau)\times_{\mathcal Q(\tilde\tau)}\left(\Delta_{\mathcal Q(\tilde\tau)}\oplus T^\star\mathcal Q(\tilde\tau)\right).    
\end{equation*}
Then the reduced path is given by
\begin{equation*}
[\texttt{c}]_G\simeq\big((\alpha_T,\alpha_\Sigma),(\nu_{\alpha_\Sigma},\xi_
{\alpha_\Sigma}),(\pi_{\alpha_\Sigma},\rho_{\alpha_\Sigma})\big)\in\Omega(\Sigma,\tilde\tau)\times_{\boldsymbol\Sigma(\tilde\tau)}\Big(\big(\Delta_{\boldsymbol\Sigma(\tilde\tau)}\oplus\Delta_{\tilde{\boldsymbol{\mathfrak g}}(\tilde\tau)}\big)\oplus\big(T^\star\boldsymbol\Sigma(\tilde\tau)\oplus\tilde{\boldsymbol{\mathfrak g}}^\star(\tilde\tau)\big)\Big),  
\end{equation*}
where $\alpha_\Sigma=[\alpha_Q]_G$, $(\nu_{\alpha_\Sigma},\xi_{\alpha_\Sigma})\simeq[\nu_{\alpha_Q}]_G$ and $(\pi_{\alpha_\Sigma},\rho_{\alpha_\Sigma})\simeq[\pi_{\alpha_Q}]_G$.
\end{lemma}

Note that the reduced path may be split into the \emph{horizontal reduced path},
\begin{equation*}
[\texttt{c}]_G^h\simeq\big((\alpha_T,\alpha_\Sigma),\nu_{\alpha_\Sigma},\pi_{\alpha_\Sigma},\big)\in\Omega(\Sigma,\tilde\tau)\times_{\boldsymbol\Sigma(\tilde\tau)}\big(\Delta_{\boldsymbol\Sigma(\tilde\tau)}\oplus T^\star\boldsymbol\Sigma(\tilde\tau)\big),
\end{equation*}
and the \emph{vertical reduced path},
\begin{equation*}
[\texttt{c}]_G^v\simeq\big((\alpha_T,\alpha_\Sigma),\xi_
{\alpha_\Sigma},\rho_{\alpha_\Sigma}\big)\in\Omega(\Sigma,\tilde\tau)\times_{\boldsymbol\Sigma(\tilde\tau)}\big(\Delta_{\tilde{\boldsymbol{\mathfrak g}}(\tilde\tau)}\oplus\tilde{\boldsymbol{\mathfrak g}}^\star(\tilde\tau)\big),
\end{equation*}

At last, given $\alpha_\Sigma\in\boldsymbol\Sigma(\tilde\tau)$, $\nu_{\alpha_\Sigma}\in T_{\alpha_\Sigma}\boldsymbol\Sigma(\tilde\tau)$ and $\pi_{\alpha_\Sigma}\in T_{\alpha_\Sigma}^\star\boldsymbol\Sigma(\tilde\tau)$, the spaces $T_{\nu_{\alpha_\Sigma}}(T\boldsymbol\Sigma(\tilde\tau))$ and $T_{\pi_{\alpha_\Sigma}}(T^*\boldsymbol\Sigma(\tilde\tau))$ are defined in the same vein as $T_{\nu_{\alpha_Q}}(T\mathcal Q(\tilde\tau))$ and $T_{\pi_{\alpha_Q}}(T^*\mathcal Q(\tilde\tau))$ by exchanging $Q$ by $\Sigma$. In addition, for each $\xi_{\alpha_\Sigma}\in\tilde{\boldsymbol{\mathfrak g}}(\tilde\tau)_{\alpha_\Sigma}$, we denote
\begin{multline*}
\qquad\qquad T_{\xi_{\alpha_\Sigma}}\tilde{\boldsymbol{\mathfrak g}}(\tilde\tau)=\Big\{\delta\xi_{\alpha_\Sigma}\in C^0([\tau_0,\tau_1],T\tilde{\mathfrak g})\mid\xi_{\alpha_\Sigma}=\pi_{T\tilde{\mathfrak g}}\circ\delta\xi_{\alpha_\Sigma},\\
\delta\xi_{\alpha_\Sigma}\text{ is piecewise }C^2\text{ and has only one singularity at }\tilde\tau\Big\},\qquad\qquad
\end{multline*}
where $\pi_{T\tilde{\mathfrak g}}:T\tilde{\mathfrak g}\to\tilde{\mathfrak g}$ is the natural projection. Analogously, we denote $T_{\rho_{\alpha_\Sigma}}\tilde{\boldsymbol{\mathfrak g}}^\star(\tilde\tau)$ for each $\rho_{\alpha_\Sigma}\in \tilde{\boldsymbol{\mathfrak g}}^\star(\tilde\tau)_{\alpha_\Sigma}$. In order to incorporate the constraint distribution, we set $\Delta_{T\boldsymbol\Sigma(\tilde\tau)}(\nu_{\alpha_\Sigma})$ as in \eqref{eq:DeltaTQ} by exchanging $Q$ by $\Sigma$.


\subsection{Nonholonomic implicit Lagrange--Poincaré equations with collisions}

Let $L:TQ\to\mathbb R$ be a (possibly degenerate) $G$-invariant Lagrangian, i.e.,
\begin{equation*}
L\left(v_q\right)=L\left((d\Phi_g)_h(v_q)\right),\qquad g\in G,~ q\in Q,~v_q\in T_q Q,
\end{equation*}
The \emph{reduced} or \emph{dropped Lagrangian} is defined as
\begin{equation*}
\ell:(TQ)/G\simeq T\Sigma\oplus\tilde{\mathfrak g}\to\mathbb R,\quad[v_q]_G\simeq(v_\sigma,\xi_\sigma)\mapsto\ell(v_\sigma,\xi_\sigma)=L(v_q),
\end{equation*}
where identification \eqref{eq:TQG} has been used. As the energy of the system \eqref{eq:energy} is also $G$-invariant, the \emph{reduced energy} may be defined as
\begin{equation*}
e:(T\Sigma\times\tilde{\mathfrak g})\oplus(T^*\Sigma\times\tilde{\mathfrak g}^*)\to\mathbb R,\quad((v_\sigma,\xi_\sigma),(p_\sigma,\rho_\sigma))\mapsto p_\sigma\cdot v_\sigma+\rho_\sigma\cdot\xi_\sigma-\ell(v_\sigma,\xi_\sigma),
\end{equation*}
where the identification $(TQ\oplus T^*Q)/G\simeq(T\Sigma\times\tilde{\mathfrak g})\oplus(T^*\Sigma\times\tilde{\mathfrak g}^*)$ has been used.

\subsubsection{Reduced variational principle}

Given a path $\tilde\alpha=(\alpha_T,\alpha_\Sigma)\in\Omega(\Sigma,\tilde\tau)$, its \emph{associated curve} is defined as
\begin{equation*}
\sigma_{\tilde\alpha}=\alpha_\Sigma\circ\alpha_T^{-1}:[t_0,t_1]\to\Sigma.
\end{equation*}
Additionally, for the paths $(\nu_{\alpha_\Sigma},\xi_{\alpha_\Sigma})\in T_{\alpha_\Sigma}\boldsymbol\Sigma(\tilde\tau)\oplus\tilde{\boldsymbol{\mathfrak g}}(\tilde\tau)_{\alpha_\Sigma}$ and $(\pi_{\alpha_\Sigma},\rho_{\alpha_\Sigma})\in T_{\alpha_\Sigma}^\star\boldsymbol\Sigma(\tilde\tau)\oplus\tilde{\boldsymbol{\mathfrak g}}^\star(\tilde\tau)_{\alpha_\Sigma}$, we set
\begin{align*}
& u_{\tilde\alpha}=\nu_{\alpha_\Sigma}\circ\alpha_T^{-1}:[t_0,t_1]\to T\Sigma,\qquad && y_{\tilde\alpha}=\pi_{\alpha_\Sigma}\circ\alpha_T^{-1}:[t_0,t_1]\to T^*\Sigma,\\
& \overline\xi_{\tilde\alpha}=\xi_{\alpha_\Sigma}\circ\alpha_T^{-1}:[t_0,t_1]\to\tilde{\mathfrak g},\qquad && \overline\rho_{\tilde\alpha}=\rho_{\alpha_\Sigma}\circ\alpha_T^{-1}:[t_0,t_1]\to\tilde{\mathfrak g}^*.
\end{align*}
Observe that $\pi_{T\Sigma}\circ u_{\tilde\alpha}=\pi_{\tilde{\mathfrak g}}\circ\overline\xi_{\tilde\alpha}=\pi_{T^*\Sigma}\circ y_{\tilde\alpha}=\pi_{\tilde{\mathfrak g}^*}\circ\overline\rho_{\tilde\alpha}=\sigma_{\tilde\alpha}$. 

\begin{remark}\label{remark:etaalpha}
As we will see below, given $\alpha=(\alpha_T,\alpha_Q)\in\Omega(Q,\tilde\tau)$ and the corresponding reduced path $\tilde\alpha=(\alpha_T,\alpha_\Sigma=[\alpha_Q]_G)\in\Omega(\Sigma,\tilde\tau)$, we set $\eta_{\alpha_\Sigma}=[\alpha_Q,\omega\circ\alpha_Q']_G$. The change of variable $t=\alpha_T(\tau)$ yields $\dot q_\alpha=\alpha_Q'/\alpha_T'$, which together with \eqref{eq:associatedcurve}, yield
\begin{equation*}
\overline\eta_{\tilde\alpha}=\eta_{\alpha_\Sigma}\circ\alpha_T^{-1}=(\alpha_T'\circ\alpha_T^{-1})\,[q_\alpha,\omega\circ\dot q_\alpha]_G.
\end{equation*}
\end{remark}

By using the reduced Lagrangian, we may define the \emph{reduced action functional},
\begin{equation*}
\mathfrak S:\Omega(\Sigma,\tilde\tau)\times_{\boldsymbol\Sigma(\tilde\tau)}\Big(\big(\Delta_{\boldsymbol\Sigma(\tilde\tau)}\oplus\Delta_{\tilde{\boldsymbol{\mathfrak g}}(\tilde\tau)}\big)\oplus\big(T^\star\boldsymbol\Sigma(\tilde\tau)\oplus\tilde{\boldsymbol{\mathfrak g}}^\star(\tilde\tau)\big)\Big)\times_{\boldsymbol\Sigma(\tilde\tau)}\tilde{\boldsymbol{\mathfrak g}}(\tilde\tau)\to\mathbb R,
\end{equation*}
as follows
\begin{align*}
& \mathfrak S\big((\alpha_T,\alpha_\Sigma),(\nu_{\alpha_\Sigma},\xi_
{\alpha_\Sigma}),(\pi_{\alpha_\Sigma},\rho_{\alpha_\Sigma}),\eta_{\alpha_\Sigma}\big)\\
& =\int_{t_0}^{t_1}\left(\ell\left(u_{\tilde\alpha}(t),\overline\xi_{\tilde\alpha}(t)\right)+y_{\tilde\alpha}(t)\cdot\left(\dot\sigma_{\tilde\alpha}(t)-u_{\tilde\alpha}(t)\right)+\overline\rho_{\tilde\alpha}(t)\cdot\left(\frac{\overline\eta_{\tilde\alpha}(t)}{(\alpha_T'\circ\alpha_T^{-1})(t)}-\overline\xi_{\tilde\alpha}(t)\right)\right)dt\\
& =\int_{\tau_0}^{\tau_1}\bigg(\ell\left(\nu_{\alpha_\Sigma}(\tau),\xi_{\alpha_\Sigma}(\tau)\right)+\pi_{\alpha_\Sigma}(\tau)\cdot\left(\frac{\alpha_\Sigma'(\tau)}{\alpha_T'(\tau)}-\nu_{\alpha_\Sigma}(\tau)\right)\\
& \hspace{15mm}+\rho_{\alpha_\Sigma(\tau)}\cdot\left(\frac{\eta_{\alpha_\Sigma}(\tau)}{\alpha_T'(\tau)}-\xi_{\alpha_\Sigma}(\tau)\right)\bigg)\,\alpha_T'(\tau)\,d\tau.
\end{align*}
As for the unreduced case, the equality for the first and the second expressions comes from the change of variables $t=\alpha_T(\tau)$. Furthermore, the reduced action functional may be written in terms of the reduced energy,
\begin{align*}
& \mathfrak S\big((\alpha_T,\alpha_\Sigma),(\nu_{\alpha_\Sigma},\xi_
{\alpha_\Sigma}),(\pi_{\alpha_\Sigma},\rho_{\alpha_\Sigma}),\eta_{\alpha_\Sigma}\big)\\
& =\int_{t_0}^{t_1}\left(y_{\tilde\alpha}(t)\cdot\dot\sigma_{\tilde\alpha}(t)+\overline\rho_{\tilde\alpha}(t)\cdot\frac{\overline\eta_{\tilde\alpha}(t)}{(\alpha_T'\circ\alpha_T^{-1})(t)}-e\left(\left(u_{\tilde\alpha}(t),\overline\xi_{\tilde\alpha}(t)\right),\left(y_{\tilde\alpha}(t),\overline\rho_{\tilde\alpha}(t)\right)\right)\right)dt\\
& =\int_{\tau_0}^{\tau_1}\bigg(\pi_{\alpha_\Sigma}(\tau)\cdot\frac{\alpha_\Sigma'(\tau)}{\alpha_T'(\tau)}+\rho_{\alpha_\Sigma(\tau)}\cdot\frac{\eta_{\alpha_\Sigma}(\tau)}{\alpha_T'(\tau)}-e\left(\left(\nu_{\alpha_\Sigma}(\tau),\xi_{\alpha_\Sigma}(\tau)\right),\left(\pi_{\alpha_\Sigma}(\tau),\rho_{\alpha_\Sigma}(\tau)\right)\right)\bigg)\,\alpha_T'(\tau)\,d\tau.
\end{align*}

In order to reduce the variational principle and obtain the reduced equations we need to introduce linear connections on the reduced spaces. More specifically, recall that $\omega$ induces a linear connection on the adjoint bundle, $\nabla^\omega:\Gamma(\tilde{\mathfrak g})\to\Omega^1(\Sigma,\tilde{\mathfrak g})$. The dual of $\nabla^\omega$, which is a linear connection on the coadjoint bundle, is denoted by $\nabla^{\omega*}:\Gamma(\tilde{\mathfrak g}^*)\to\Omega^1(\Sigma,\tilde{\mathfrak g}^*)$. We also consider a linear connection on the tangent bundle of $\Sigma$, $\nabla^\Sigma:\mathfrak X(\Sigma)\to\Omega^1(\Sigma,T\Sigma)$, as well as its dual, which is a linear connection on the cotangent bundle of $\Sigma$, $\nabla^{\Sigma*}:\Omega^1(\Sigma)\to\Omega^1(\Sigma,T^*\Sigma)$. These connections allow for computing the vertical parts of the variations. For instance, given $\alpha_\Sigma\in\boldsymbol\Sigma(\tilde\tau)$, $\xi_{\alpha_\Sigma}\in\tilde{\boldsymbol{\mathfrak g}}(\tilde\tau)_{\alpha_\Sigma}$ and $\delta\xi_{\alpha_\Sigma}\in T_{\xi_{\alpha_\Sigma}}\tilde{\boldsymbol{\mathfrak g}}(\tilde\tau)$, we denote 
\begin{equation*}
\delta^\omega\xi_{\alpha_\Sigma}=\left(\delta\xi_{\alpha_\Sigma}\right)^v\in\tilde{\boldsymbol{\mathfrak g}}(\tilde\tau)_{\alpha_\Sigma},
\end{equation*}
where the superscript $v$ denotes the vertical part computed with $\nabla^\omega$, and analogous for the other connections: $\delta^\Sigma$, $\delta^{\omega*}$ and $\delta^{\Sigma*}$. The covariant derivatives induced by these connection are denoted by $\nabla^\omega/d\tau$, $\nabla^\Sigma/d\tau$, $\nabla^{\omega*}/d\tau$ and $\nabla^{\Sigma*}/d\tau$, respectively. 

In the following, the (fiberwise) adjoint representation of $\tilde{\mathfrak g}$ is denoted by $\ad_{\xi_\sigma}:\tilde{\mathfrak g}_\sigma\to\tilde{\mathfrak g}_\sigma$ for each $\xi_\sigma\in\tilde{\mathfrak g}_\sigma$, $\sigma\in\Sigma$. Likewise, the (fiberwise) coadjoint representation, i.e., minus the adjoint of the adjoint representation, is denoted by $\ad_{\xi_\sigma}^*:\tilde{\mathfrak g}_\sigma^*\to\tilde{\mathfrak g}_\sigma^*$. In addition, the reduced curvature of the principal connection $\omega\in\Omega^1(Q,\mathfrak g)$ is denoted by $\tilde F^\omega\in\Omega^2(\Sigma,\tilde{\mathfrak g})$. 

\begin{theorem}[Reduced Hamilton--d'Alembert--Pontryagin principle]\label{theorem:reducedHdAPprinciple}
Let
\begin{equation*}
\texttt{c}=\big((\alpha_T,\alpha_Q),\nu_{\alpha_Q},\pi_{\alpha_Q}\big)\in\Omega(Q,\tilde\tau)\times_{\mathcal Q(\tilde\tau)}\left(\Delta_{\mathcal Q(\tilde\tau)}\oplus T^\star\mathcal Q(\tilde\tau)\right)    
\end{equation*}
be a path,
\begin{equation*}
[\texttt{c}]_G\simeq\big((\alpha_T,\alpha_\Sigma),(\nu_{\alpha_\Sigma},\xi_
{\alpha_\Sigma}),(\pi_{\alpha_\Sigma},\rho_{\alpha_\Sigma})\big)\in\Omega(\Sigma,\tilde\tau)\times_{\boldsymbol\Sigma(\tilde\tau)}\Big(\big(\Delta_{\boldsymbol\Sigma(\tilde\tau)}\oplus\Delta_{\tilde{\boldsymbol{\mathfrak g}}(\tilde\tau)}\big)\oplus\big(T^\star\boldsymbol\Sigma(\tilde\tau)\oplus\tilde{\boldsymbol{\mathfrak g}}^\star(\tilde\tau)\big)\Big),  
\end{equation*}
be the corresponding reduced path, as in Lemma \ref{lemma:reducedpath}, and
\begin{equation*}
\eta_{\alpha_\Sigma}=[\alpha_Q,\omega\circ\alpha_Q']_G\in\tilde{\boldsymbol{\mathfrak g}}(\tilde\tau)_{\alpha_\Sigma}.
\end{equation*}
Then $\texttt{c}$ is critical for the action functional $\mathbb S$ (recall Definition \ref{def:HdAPprinciple}) if and only if
\begin{equation*}
d\mathfrak S\big([\texttt{c}]_G,\eta_{\alpha_\Sigma}\big)\big(\delta[\texttt{c}]_G,\delta\eta_{\alpha_\Sigma}\big)=0,
\end{equation*}
for every variation
\begin{multline*}
\delta[\texttt{c}]_G=\big((\delta\alpha_T,\delta\alpha_\Sigma),(\delta\nu_{\alpha_\Sigma},\delta\xi_{\alpha_\Sigma}),(\delta\pi_{\alpha_\Sigma},\delta\rho_{\alpha_\Sigma})\big)\\
\in T_{(\alpha_T,\alpha_\Sigma)}\Omega(\Sigma,\tilde\tau)\times\left(\Delta_{T\boldsymbol\Sigma(\tilde\tau)}(\nu_{\alpha_\Sigma})\oplus T_{\xi_{\alpha_\Sigma}}\tilde{\boldsymbol{\mathfrak g}}(\tilde\tau)\right)\times\left(T_{\pi_{\alpha_\Sigma}}(T^\star\boldsymbol\Sigma(\tilde\tau))\oplus T_{\rho_{\alpha_\Sigma}}\tilde{\boldsymbol{\mathfrak g}}^\star(\tilde\tau)\right),
\end{multline*}
such that $\delta\alpha_T(\tau_0)=\delta\alpha_T(\tau_1)=0$, $\delta\alpha_\Sigma(\tau_0)=\delta\alpha_\Sigma(\tau_1)=0$ and
\begin{equation}\label{eq:projectionvariationreduced}
d\pi_{T\Sigma}\circ\delta\nu_{\alpha_\Sigma}= d\pi_{\tilde{\mathfrak g}}\circ\delta\xi_{\alpha_\Sigma}=d\pi_{T^*\Sigma}\circ\delta\pi_{\alpha_\Sigma}=d\pi_{\tilde{\mathfrak g}^*}\circ\rho_{\alpha_\Sigma}=\delta\alpha_\Sigma,
\end{equation}
and for every variation $\delta\eta_{\alpha_\Sigma}\in T_{\eta_{\alpha_\Sigma}}\tilde{\boldsymbol{\mathfrak g}}(\tilde\tau)$ such that
\begin{equation}\label{eq:deltaeta}
\delta^\omega\eta_{\alpha_\Sigma}=\frac{\nabla^\omega\hat\eta_{\alpha_\Sigma}}{d\tau}+\ad_{\eta_{\alpha_\Sigma}}(\hat\eta_{\alpha_\Sigma})+\tilde F^\omega(\delta\alpha_\Sigma,\alpha_\Sigma')\in\tilde{\boldsymbol{\mathfrak g}}(\tilde\tau)_{\alpha_\Sigma},
\end{equation}
where $\nabla^\omega/d\tau$ denotes the covariant derivative induced by $\nabla^\omega$, for arbitrary paths $\hat\eta_{\alpha_\Sigma}\in\tilde{\boldsymbol{\mathfrak g}}(\tilde\tau)_{\alpha_\Sigma}$ such that $\hat\eta_{\alpha_\Sigma}(\tau_0)=\hat\eta_{\alpha_\Sigma}(\tau_1)=0$.
\end{theorem}

\begin{proof}
To begin with, note that $\mathbb S(\texttt{c})=\mathfrak S([\texttt{c}]_G,\eta_{\alpha_\Sigma})$ by construction. Now let
\begin{equation*}
\left\{\texttt{c}^s=\left((\alpha_T^s,\alpha_Q^s),\nu_{\alpha_Q}^s,\pi_{\alpha_Q}^s\right)\in\Omega(Q,\tilde\tau)\times_{\mathcal Q(\tilde\tau)}\left(\Delta_{\mathcal Q(\tilde\tau)}\oplus T^\star\mathcal Q(\tilde\tau)\right)|s\in(-\epsilon,\epsilon)\right\}
\end{equation*}
be a variation of $\texttt{c}$ such that $\delta\texttt{c}=d/ds|_{s=0}\texttt{c}^s=\left((\delta\alpha_T,\delta\alpha_Q),\delta\nu_{\alpha_Q},\delta\pi_{\alpha_Q}\right)$. Then the reduced variation is computed from \eqref{eq:TQG}, leading to
\begin{align*}
\delta[\texttt{c}]_G & =\left.\frac{d}{ds}\right|_{s=0}[\texttt{c}^s]_G\\
& =\left.\frac{d}{ds}\right|_{s=0}\big((\alpha_T^s,\alpha_\Sigma^s),(\nu_{\alpha_\Sigma}^s,\xi_
{\alpha_\Sigma}^s),(\pi_{\alpha_\Sigma}^s,\rho_{\alpha_\Sigma}^s)\big)\\
& =\left.\frac{d}{ds}\right|_{s=0}\big((\alpha_T^s,\pi_{Q,\Sigma}\circ\alpha_Q^s),(d\pi_{Q,\Sigma}\circ\nu_{\alpha_Q}^s,[\alpha_Q^s,\omega\circ\nu_{\alpha_Q}^s]_G),(\texttt H^*\circ\pi_{\alpha_Q}^s,[\alpha_Q^s,\mathbf J\circ\pi_{\alpha_Q}^s]_G)\big)\\
& =\big((\delta\alpha_T,\delta\alpha_\Sigma),(\delta\nu_{\alpha_\Sigma},\delta\xi_{\alpha_\Sigma}),(\delta\pi_{\alpha_\Sigma},\delta\rho_{\alpha_\Sigma})\big),
\end{align*}
where we have denoted $\delta\alpha_\Sigma=d\pi_{Q,\Sigma}\circ\delta\alpha_Q$, $\delta\nu_{\alpha_\Sigma}=d(d\pi_{Q,\Sigma})\circ\delta\nu_{\alpha_Q}$, $\delta\pi_{\alpha_\Sigma}=d\texttt H^*\circ\delta\pi_{\alpha_Q}$ and
\begin{equation*}
\delta\xi_{\alpha_\Sigma}=\left.\frac{d}{ds}\right|_{s=0}[\alpha_Q^s,\omega\circ\nu_{\alpha_Q}^s]_G,\qquad\delta\rho_{\alpha_\Sigma}=\left.\frac{d}{ds}\right|_{s=0}[\alpha_Q^s,\mathbf J\circ\pi_{\alpha_Q}^s]_G.
\end{equation*}
It is now clear that free variations $\delta\alpha_Q$, $\delta\nu_{\alpha_Q}$ and $\delta\pi_{\alpha_Q}$ induce free variations $\delta\alpha_\Sigma$, $\delta\xi_{\alpha_\Sigma}$, $\delta\pi_{\alpha_\Sigma}$ and $\delta\rho_{\alpha_\Sigma}$. Note that, although we have the restriction $\delta\nu_{\alpha_Q}\in\Delta_{T\mathcal Q}(\nu_{\alpha_Q})$, the condition \eqref{eq:DeltaTQ} does not restrict $\delta\xi_{\alpha_\Sigma}$. This may be checked locally, where $\delta\nu_{\alpha_Q}=(\alpha,\nu,\delta\alpha,\delta\nu)$ and the restriction reads $\delta\alpha\in\Delta_{\mathcal Q(\tilde\tau)}(\alpha_Q)$, but $\delta\xi_{\alpha_Q}$ only depends on $\delta\nu$, which is free.

Next, we have
\begin{equation*}
\left\{\begin{array}{rl}
\displaystyle d\pi_{T\Sigma}\circ\delta\nu_{\alpha_\Sigma} & \displaystyle=\left.\frac{d}{ds}\right|_{s=0}\pi_{T\Sigma}\circ d\pi_{Q,\Sigma}\circ\nu_{\alpha_Q}^s=d\pi_{Q,\Sigma}\circ d\pi_{TQ}\circ\delta\nu_{\alpha_Q},\vspace{2mm}\\
\displaystyle d\pi_{\tilde{\mathfrak g}}\circ\delta\xi_{\alpha_\Sigma} & \displaystyle=\left.\frac{d}{ds}\right|_{s=0}\pi_{\tilde{\mathfrak g}}\circ[\alpha_Q^s,\omega\circ\nu_{\alpha_Q}^s]_G=d\pi_{Q,\Sigma}\circ\delta\alpha_Q=\delta\alpha_\Sigma,\vspace{2mm}\\
\displaystyle d\pi_{T^*\Sigma}\circ\delta\pi_{\alpha_\Sigma} & \displaystyle=\left.\frac{d}{ds}\right|_{s=0}\pi_{T^*\Sigma}\circ\texttt{H}^*\circ\pi_{\alpha_Q}^s=d\pi_{Q,\Sigma}\circ d\pi_{T^*Q}\circ\delta\pi_{\alpha_Q},\vspace{2mm}\\
\displaystyle d\pi_{\tilde{\mathfrak g}^*}\circ\rho_{\alpha_\Sigma} & \displaystyle=\left.\frac{d}{ds}\right|_{s=0}\pi_{\tilde{\mathfrak g}^*}\circ[\alpha_Q^s,\mathbf J\circ\nu_{\alpha_Q}^s]_G=d\pi_{Q,\Sigma}\circ\delta\alpha_Q=\delta\alpha_\Sigma,
\end{array}\right.
\end{equation*}
where we have used that $\pi_{T\Sigma}\circ d\pi_{Q,\Sigma}=\pi_{Q,\Sigma}\circ\pi_{TQ}$ and $\pi_{T^*\Sigma}\circ\texttt{H}^*=\pi_{Q,\Sigma}\circ\pi_{T^*Q}$. From the first equality, we conclude that the condition $\delta\nu_{\alpha_Q}\in\Delta_{T\mathcal Q(\tilde\tau)}(\nu_{\alpha_Q})$ yields $\delta\nu_{\alpha_\Sigma}\in\Delta_{T\boldsymbol\Sigma(\tilde\tau)}(\nu_{\alpha_\Sigma})$, where we have used \eqref{eq:TQG} and \eqref{eq:reducedDeltaQ}. Moreover, a straightforward check shows that the condition \eqref{eq:projectionvariation} reduces to \eqref{eq:projectionvariationreduced}.

The second part is a straightforward consequence of \cite[Theorem 3.3.1]{CeMaRa2001}, which can be applied independently on $[\tau_0,\tilde\tau)$ and $(\tilde\tau,\tau_1]$. Recall that in our case the parameter of the paths is $\tau$ instead of $t$.
\end{proof}

\subsubsection{Reduced equations}

The fiber derivatives of the reduced Lagrangian are denoted by
\begin{equation*}
\begin{array}{ll}
\displaystyle\frac{\delta\ell}{\delta v_\sigma}:T\Sigma\oplus\tilde{\mathfrak g}\to T^*\Sigma,\qquad & \displaystyle\frac{\delta\ell}{\delta v_\sigma}(v_\sigma,\xi_\sigma)\cdot w_\sigma=\left.\frac{d}{ds}\right|_{s=0}\ell(v_\sigma+s\,w_\sigma,\xi_\sigma),\vspace{2mm}\\
\displaystyle\frac{\delta\ell}{\delta\xi_\sigma}:T\Sigma\oplus\tilde{\mathfrak g}\to\tilde{\mathfrak g}^*,\qquad & \displaystyle\frac{\delta\ell}{\delta\xi_\sigma}(v_\sigma,\xi_\sigma)\cdot\eta_\sigma=\left.\frac{d}{ds}\right|_{s=0}\ell(v_\sigma,\xi_\sigma+s\,\eta_\sigma),
\end{array}
\end{equation*}
for each $(v_\sigma,\xi_\sigma),(w_\sigma,\eta_\sigma)\in T_\sigma\Sigma\oplus\tilde{\mathfrak g}_\sigma$, $\sigma\in\Sigma$. Moreover, the horizontal derivative of the reduced Lagrangian is defined by using the linear connection $\nabla^\Sigma\oplus\nabla^\omega$ on $T\Sigma\oplus\tilde{\mathfrak g}$ as
\begin{equation*}
\frac{\delta\ell}{\delta\sigma}:T\Sigma\oplus\tilde{\mathfrak g}\to T^*\Sigma,\qquad\frac{\delta\ell}{\delta\sigma}(v_\sigma,\xi_\sigma)\cdot w_\sigma=\left.\frac{d}{ds}\right|_{s=0}\left(\ell\circ\gamma_{(v_\sigma,\xi_\sigma)}^h\right)(s),
\end{equation*}
for each $(v_\sigma,\xi_\sigma)\in T_\sigma\Sigma\oplus\tilde{\mathfrak g}_\sigma$ and $w_\sigma\in T_\sigma\Sigma$, where $\gamma:(-\epsilon,\epsilon)\to\Sigma$ is a curve such that $d/ds|_{s=0}\gamma(s)=w_\sigma$ and $\gamma_{(v_\sigma,\xi_\sigma)}^h:(-\epsilon,\epsilon)\to T\Sigma\oplus\tilde{\mathfrak g}$ is its horizontal lift at $(v_\sigma,\xi_\sigma)$ given by the linear connection $\nabla^\Sigma\oplus\nabla^\omega$.

\begin{theorem}[Reduced nonholonomic implicit Euler--Lagrange equations with collisions]
Let $\texttt{c}=((\alpha_T,\alpha_Q),\nu_{\alpha_Q},\pi_{\alpha_Q})\in\Omega(Q,\tilde\tau)\times_{\mathcal Q(\tilde\tau)}\left(\Delta_{\mathcal Q(\tilde\tau)}\oplus T^\star\mathcal Q(\tilde\tau)\right)$ be a path and
\begin{equation*}
[\texttt{c}]_G=\big((\alpha_T,\alpha_\Sigma),(\nu_{\alpha_\Sigma},\xi_
{\alpha_\Sigma}),(\pi_{\alpha_\Sigma},\rho_{\alpha_\Sigma})\big)\in\Omega(\Sigma,\tilde\tau)\times_{\boldsymbol\Sigma(\tilde\tau)}\Big(\big(\Delta_{\boldsymbol\Sigma(\tilde\tau)}\oplus\Delta_{\tilde{\boldsymbol{\mathfrak g}}(\tilde\tau)}\big)\oplus\big(T^\star\boldsymbol\Sigma(\tilde\tau)\oplus\tilde{\boldsymbol{\mathfrak g}}^\star(\tilde\tau)\big)\Big)
\end{equation*}
be the reduced path (recall Lemma \ref{lemma:reducedpath}). Then the following statements are equivalent
\begin{enumerate}
    \item $\texttt{c}$ is stationary for the Hamilton--d'Alembert--Pontryagin variational principle introduced in Definition \ref{def:HdAPprinciple}.
    \item $\texttt{c}$ satisfies the nonholonomic implicit Euler--Lagrange equations with collisions given in Theorem \ref{theorem:implicitELequations}.
    \item $[\texttt{c}]_G$ and $\eta_{\alpha_\Sigma}=[\alpha_Q,\omega\circ\alpha_Q']_G$ satisfy the reduced Hamilton--d'Alembert--Pontryagin variational principle given in Theorem \ref{theorem:reducedHdAPprinciple}.
    \item $[\texttt{c}]_G$ and $\eta_{\alpha_\Sigma}=[\alpha_Q,\omega\circ\alpha_Q']_G$ satisfy the \emph{nonholonomic implicit Lagrange--Poincaré equations with collisions}, which consists of the \emph{horizontal equations} on  $[\tau_0,\tilde\tau)\cup(\tilde\tau,\tau_1]$,
    \begin{equation*}
    \left\{\begin{array}{l}
    \displaystyle\frac{\nabla^{\Sigma*}\pi_{\alpha_\Sigma}}{d\tau}-\frac{\delta\ell}{\delta\sigma}(\nu_{\alpha_\Sigma},\xi_{\alpha_\Sigma})\alpha_T'+\rho_{\alpha_\Sigma}\cdot\left(i_{\alpha_\Sigma'}\tilde F^\omega\right)\in\Delta_{\Sigma}^\circ(\alpha_\Sigma),\vspace{2mm}\\
    \displaystyle\frac{\alpha_\Sigma'}{\alpha_T'}=\nu_{\alpha_\Sigma}\in\Delta_\Sigma(\alpha_\Sigma),\qquad \frac{\delta\ell}{\delta v_\sigma}(\nu_{\alpha_\Sigma},\xi_{\alpha_\Sigma})=\pi_{\alpha_\Sigma},
    \end{array}\right.
    \end{equation*}
    where $i_U:\Omega^2(\Sigma,\tilde{\mathfrak g})\to\Omega^1(\Sigma,\tilde{\mathfrak g})$ denotes the left interior product by $U\in\mathfrak X(\Sigma)$, the \emph{vertical equations} on $[\tau_0,\tilde\tau)\cup(\tilde\tau,\tau_1]$,
    \begin{equation*}
    \left\{\begin{array}{l}
    \displaystyle\frac{\nabla^{\omega*}\rho_{\alpha_\Sigma}}{d\tau}=\ad_{\eta_{\alpha_\Sigma}}^*(\rho_{\alpha_\Sigma}),\vspace{2mm}\\
    \displaystyle\frac{\eta_{\alpha_\Sigma}}{\alpha_T'}=\xi_{\alpha_\Sigma}\in\Delta_{\tilde{\mathfrak g}}(\alpha_\Sigma),\qquad\frac{\delta\ell}{\delta\xi_\sigma}(\nu_{\alpha_\Sigma},\xi_{\alpha_\Sigma})=\rho_{\alpha_\Sigma},
    \end{array}\right.
    \end{equation*}
    and the \emph{reduced energy conservation} on $[\tau_0,\tilde\tau)\cup(\tilde\tau,\tau_1]$, 
    \begin{equation*}
    e'((\nu_{\alpha_\Sigma},\xi_{\alpha_\Sigma}),(\pi_{\alpha_\Sigma},\rho_{\alpha_\Sigma}))=0,
    \end{equation*}
    together with the \emph{reduced conditions for the elastic impact},
    \begin{equation*}
    \left\{\begin{array}{l}
    \displaystyle \pi_{\alpha_\Sigma}^+-\pi_{\alpha_\Sigma}^-\in\left(T_{\alpha_\Sigma(\tilde\tau)}\partial\Sigma\cap\Delta_\Sigma(\alpha_\Sigma(\tilde\tau)\right)^\circ=\left(T_{\alpha_\Sigma(\tilde\tau)}\partial\Sigma\right)^\circ+\Delta_\Sigma^\circ(\alpha_\Sigma(\tilde\tau)), \vspace{2mm}\\
    \displaystyle e\big((\nu_{\alpha_\Sigma}^+,\xi_{\alpha_\Sigma}^+),(\pi_{\alpha_\Sigma}^+,\rho_{\alpha_\Sigma}^+)\big)=e\big((\nu_{\alpha_\Sigma}^-,\xi_{\alpha_\Sigma}^-),(\pi_{\alpha_\Sigma}^-,\rho_{\alpha_\Sigma}^-)\big),
    \end{array}\right.
    \end{equation*}
    where we have denoted $\pi_{\alpha_\Sigma}(\tilde\tau^+)=\pi_{\alpha_\Sigma}^+$, etc.
\end{enumerate}
\end{theorem}

\begin{proof}
The equivalences between $(i)$ and $(ii)$, and between $(i)$ and $(iii)$ were established in Theorems \ref{theorem:implicitELequations} and \ref{theorem:reducedHdAPprinciple}. To conclude, let us show the equivalence between $(iii)$ and $(iv)$. To that end, Theorem \ref{theorem:reducedHdAPprinciple} is used and the reduced variations are decomposed into its horizontal and vertical parts by means of the linear connections $\nabla^\omega$, $\nabla^\Sigma$, $\nabla^{\omega*}$ and $\nabla^{\Sigma*}$. Note that \eqref{eq:projectionvariationreduced} ensures that the horizontal part of all the variations is $\delta\alpha_\Sigma\in\Delta_{\boldsymbol\Sigma(\tilde\tau)}(\alpha_\Sigma)$. Therefore,
\begin{align*}
& d\mathfrak S\left([\texttt c]_G,\eta_{\alpha_\Sigma}\right)\left(\delta[\texttt c ]_G,\delta\eta_{\alpha_\Sigma}\right)\\
& \hspace{10mm}=\int_{\tau_0}^{\tau_1}\bigg(\frac{\delta\ell}{\delta\sigma}\cdot\delta\alpha_\Sigma+\frac{\delta\ell}{\delta v_\sigma}\cdot\delta^\Sigma\nu_{\alpha_\Sigma}+\frac{\delta\ell}{\delta\xi_\sigma}\cdot\delta^\omega\xi_{\alpha_\Sigma}+\delta^{\Sigma*}\pi_{\alpha_\Sigma}\cdot\bigg(\frac{\alpha_\Sigma'}{\alpha_T'}-\nu_{\alpha_\Sigma}\bigg)\\
& \hspace{25mm}+\pi_{\alpha_\Sigma}\cdot\bigg(\frac{\delta^\Sigma\alpha_\Sigma'}{\alpha_T'}-\frac{\alpha_\Sigma'\delta\alpha_T'}{(\alpha_T')^2}-\delta^\Sigma\nu_{\alpha_\Sigma}\bigg)+\delta^{\Sigma*}\rho_{\alpha_\Sigma}\cdot\bigg(\frac{\eta_{\alpha_\Sigma}}{\alpha_T'}-\xi_{\alpha_\Sigma}\bigg)\\
& \hspace{25mm}+\rho_{\alpha_\Sigma}\cdot\bigg(\frac{\delta^\omega\eta_{\alpha_\Sigma}}{\alpha_T'}-\frac{\eta_{\alpha_\Sigma}\delta\alpha_T'}{(\alpha_T')^2}-\delta^\omega\xi_{\alpha_\Sigma}\bigg)\bigg)\,\alpha_T'\,d\tau\\
& \hspace{10mm}+\int_{\tau_0}^{\tau_1}\bigg(\ell+\pi_{\alpha_\Sigma}\cdot\left(\frac{\alpha_\Sigma'}{\alpha_T'}-\nu_{\alpha_\Sigma}\right)+\rho_{\alpha_\Sigma}\cdot\left(\frac{\eta_{\alpha_\Sigma}}{\alpha_T'}-\xi_{\alpha_\Sigma}\right)\bigg)\,\delta\alpha_T'\,d\tau,
\end{align*}
where the reduced Lagrangian, as well as its partial derivatives, are evaluated at \newline$(\nu_{\alpha_\Sigma}(\tau),\xi_{\alpha_\Sigma}(\tau))$.

Note that
\begin{equation*}
\delta^\Sigma\alpha_\Sigma'=\frac{\nabla^{\Sigma}\delta\alpha_\Sigma}{d\tau},\qquad(\pi_{\alpha_\Sigma}\cdot\delta\alpha_\Sigma)'=\left(\frac{\nabla^{\Sigma*}\pi_{\alpha_\Sigma}}{d\tau}\right)\cdot\delta\alpha_\Sigma+\pi_{\alpha_\Sigma}\cdot\frac{\nabla^\Sigma\delta\alpha_\Sigma}{d\tau}.
\end{equation*}
By using this, splitting the integration domain, $[\tau_1,\tilde\tau)\cup(\tilde\tau,\tau_1]$, integrating by parts on each sub-interval and regrouping terms, the previous expression leads to
\begin{equation*}
d\mathfrak S\left([\texttt c]_G,\eta_{\alpha_\Sigma}\right)\left(\delta[\texttt c]_G,\delta\eta_{\alpha_\Sigma}\right)=\mathcal I(\tau_0,\tilde\tau)+\mathcal I(\tilde\tau,\tau_1)+\mathcal B(\tau_0,\tilde\tau^-)+\mathcal B(\tilde\tau^+,\tau_1),
\end{equation*}
where for each $a,b\in\mathbb R$, $a<b$, we set
\begin{align*}
& \mathcal I(a,b)=\int_a^b\bigg(\left(\frac{\delta\ell}{\delta\sigma}\alpha_T'-\frac{\nabla^{\Sigma*}\pi_{\alpha_\Sigma}}{d\tau}\right)\cdot\delta\alpha_\Sigma+\alpha_T'\left(\frac{\delta\ell}{\delta v_\sigma}-\pi_{\alpha_\Sigma}\right)\cdot\delta^\Sigma\nu_{\alpha_\Sigma}\\
& \hspace{15mm}+\alpha_T'\left(\frac{\delta\ell}{\delta\xi_\sigma}-\rho_{\alpha_\Sigma}\right)\cdot\delta^\omega\xi_{\alpha_\Sigma}+\delta^{\Sigma*}\pi_{\alpha_\Sigma}\cdot\bigg(\frac{\alpha_\Sigma'}{\alpha_T'}-\nu_{\alpha_\Sigma}\bigg)\,\alpha_T'\\
& \hspace{15mm}+\delta^{\omega*}\rho_{\alpha_\Sigma}\cdot\bigg(\frac{\eta_{\alpha_\Sigma}}{\alpha_T'}-\xi_{\alpha_\Sigma}\bigg)\,\alpha_T'+\rho_{\alpha_\Sigma}\cdot\delta^\omega\eta_{\alpha_\Sigma}+e'\,\delta\alpha_T\bigg)\,d\tau\\
& \mathcal B(a,b)=\bigg[\pi_{\alpha_\Sigma}(\tau)\cdot\delta\alpha_\Sigma(\tau)-e\,\delta\alpha_T(\tau)\bigg]_{\tau=a}^{\tau=b},
\end{align*}
with the reduced energy evaluated at $((\nu_{\alpha_\Sigma}(\tau),\xi_{\alpha_\Sigma}(\tau)),(\pi_{\alpha_\Sigma}(\tau),\rho_{\alpha_\Sigma}(\tau)))$. The reduced conditions for the elastic impact are straightforwardly obtained by recalling that $\delta\alpha_T$ is free, $\delta\alpha_\Sigma\in\Delta_{\boldsymbol\Sigma(\tilde\tau)}(\alpha_\Sigma)$ and both of them vanish at the endpoints.

Next, we use \eqref{eq:deltaeta}, as well as integration by parts, to write
\begin{align*}
& \mathcal I(a,b)=\int_a^b\bigg(\left(\frac{\delta\ell}{\delta\sigma}\alpha_T'-\frac{\nabla^{\Sigma*}\pi_{\alpha_\Sigma}}{d\tau}-\rho_{\alpha_\Sigma}\cdot\left(i_{\alpha_\Sigma'}\tilde F^\omega\right)\right)\cdot\delta\alpha_\Sigma+\alpha_T'\left(\frac{\delta\ell}{\delta v_\sigma}-\pi_{\alpha_\Sigma}\right)\cdot\delta^\Sigma\nu_{\alpha_\Sigma}\\
& \hspace{15mm}+\alpha_T'\left(\frac{\delta\ell}{\delta\xi_\sigma}-\rho_{\alpha_\Sigma}\right)\cdot\delta^\omega\xi_{\alpha_\Sigma}+\delta^{\Sigma*}\pi_{\alpha_\Sigma}\cdot\bigg(\frac{\alpha_\Sigma'}{\alpha_T'}-\nu_{\alpha_\Sigma}\bigg)\,\alpha_T'+\delta^{\omega*}\rho_{\alpha_\Sigma}\cdot\bigg(\frac{\eta_{\alpha_\Sigma}}{\alpha_T'}-\xi_{\alpha_\Sigma}\bigg)\,\alpha_T'\\
& \hspace{15mm}+\left(\frac{\nabla^{\omega*}\rho_{\alpha_\Sigma}}{d\tau}-\ad_{\eta_{\alpha_\Sigma}}^*(\rho_{\alpha_\Sigma})\right)\cdot\hat\eta_{\alpha_\Sigma}+e'\,\delta\alpha_T\bigg)\,d\tau
\end{align*}
The reduced equations are now straightforward by using that $\delta\alpha_\Sigma\in\Delta_{\boldsymbol\Sigma(\tilde\tau)}(\alpha_\Sigma)$, $\delta^\Sigma\nu_{\alpha_\Sigma}\in T_{\alpha_\Sigma}\boldsymbol\Sigma(\tilde\tau)$, $\delta^\omega\xi_{\alpha_\Sigma}\in\tilde{\boldsymbol{\mathfrak g}}(\tilde\tau)_{\alpha_\Sigma}$, $\delta^{\Sigma*}\pi_{\alpha_\Sigma}\in T_{\alpha_\Sigma}^\star\boldsymbol\Sigma(\tilde\tau)$, $\delta^{\omega*}\rho_{\alpha_\Sigma}\in\tilde{\boldsymbol{\mathfrak g}}^*(\tilde\tau)_{\alpha_\Sigma}$, $\hat\eta_{\alpha_\Sigma}\in\tilde{\boldsymbol{\mathfrak g}}(\tilde\tau)$ and $\delta\alpha_T\in T_{\alpha_T}\mathcal T$ are free variations.
\end{proof}

As for the original (unreduced) equations, by using the change of variable $t=\alpha_T(\tau)$, the nonholonomic implicit Lagrange--Poincaré equations for a curve
\begin{equation*}
\left(\left(u_{\tilde\alpha},\overline\xi_{\tilde\alpha}\right),\left(y_{\tilde\alpha},\overline\rho_{\tilde\alpha}\right),\overline\eta_{\tilde\alpha}\right):[t_0,t_1]\to\big(T\Sigma\oplus\tilde{\mathfrak g}\big)\oplus\big(T^*\Sigma\oplus\tilde{\mathfrak g}^*\big)\oplus\tilde{\mathfrak g}
\end{equation*}
read
\begin{equation}\label{eq:LPeqs}
\begin{array}{l}
    \left\{\begin{array}{l}
    \displaystyle\frac{\nabla^{\Sigma*}y_{\tilde\alpha}}{dt}-\frac{\delta\ell}{\delta\sigma}\left(u_{\tilde\alpha},\overline\xi_{\tilde\alpha}\right)+\overline\rho_{\tilde\alpha}\cdot\left(i_{\dot\sigma_{\tilde\alpha}}\tilde F^\omega\right)\in\Delta_{\Sigma}^\circ(\sigma_{\tilde\alpha}),\vspace{2mm}\\
    \displaystyle\dot\sigma_{\tilde\alpha}=u_{\tilde\alpha}\in\Delta_\Sigma(\sigma_{\tilde\alpha}),\qquad \frac{\delta\ell}{\delta v_\sigma}\left(u_{\tilde\alpha},\overline\xi_{\tilde\alpha}\right)=y_{\tilde\alpha},
    \end{array}\right.\vspace{2mm}\\
    \left\{\begin{array}{l}
    \displaystyle\frac{\nabla^{\omega*}\overline\rho_{\tilde\alpha}}{dt}=\ad_{\overline\zeta_{\tilde\alpha}}^*\left(\overline\rho_{\tilde\alpha}\right),\vspace{2mm}\\
    \displaystyle\overline\zeta_{\tilde\alpha}=\overline\xi_{\tilde\alpha}\in\Delta_{\tilde{\mathfrak g}}(\sigma_{\tilde\alpha}),\qquad\frac{\delta\ell}{\delta\xi_\sigma}\left(u_{\tilde\alpha},\overline\xi_{\tilde\alpha}\right)=\overline\rho_{\tilde\alpha},
    \end{array}\right.\vspace{2mm}\\
    \hspace{4.5mm} e'\left(\left(u_{\tilde\alpha},\overline\xi_{\tilde\alpha}\right),\left(y_{\tilde\alpha},\overline\rho_{\tilde\alpha}\right)\right)=0,
\end{array}
\end{equation}
where we have used Remark \ref{remark:etaalpha} and denoted $\overline\zeta_{\tilde\alpha}=\overline\eta_{\tilde\alpha}/(\alpha_T'\circ\alpha_T^{-1})$. Moreover, the reduced conditions for the elastic impact read
\begin{align}\label{eq:reducedimpactL}
& y_{\tilde\alpha}^+-y_{\tilde\alpha}^-\in\left(T_{\sigma_{\tilde\alpha}\left(\tilde t\right)}\partial\Sigma\cap\Delta_\Sigma\left(\sigma_{\tilde\alpha}\left(\tilde t\right)\right)\right)^\circ=\left(T_{\sigma_{\tilde\alpha}\left(\tilde t\right)}\partial\Sigma\right)^\circ+\Delta_\Sigma^\circ\left(\sigma_{\tilde\alpha}\left(\tilde t\right)\right),\\\nonumber
& e\left(\left(u_{\tilde\alpha}^+,\overline\xi_{\tilde\alpha}^+\right),\left(y_{\tilde\alpha}^+,\overline\rho_{\tilde\alpha}^+\right)\right)=e\left(\left(u_{\tilde\alpha}^-,\overline\xi_{\tilde\alpha}^-\right),\left(y_{\tilde\alpha}^-,\overline\rho_{\tilde\alpha}^-\right)\right),\\\label{eq:reducedimpactDelta}
& \dot\sigma_{\tilde\alpha}^+=u_{\tilde\alpha}^+\in\Delta_\Sigma\left(\sigma_{\tilde\alpha}\left(\tilde t\right)\right),\qquad\overline\zeta_{\tilde\alpha}^+=\overline\xi_{\tilde\alpha}^+\in\Delta_{\tilde{\mathfrak g}}\left(\sigma_{\tilde\alpha}\left(\tilde t\right)\right),
\end{align}
where we have denoted $y_{\tilde\alpha}\left(\tilde t^+\right)=y_{\tilde\alpha}^+$, etc.

\begin{remark}
For the unconstrained case, i.e., $\Delta_\Sigma=T\Sigma$ and $\Delta_{\tilde{\mathfrak g}}=\tilde{\mathfrak g}$, we recover the \emph{implicit Lagrange--Poincaré equations} presented in  \cite{YoMa2009}.

\end{remark}

\subsection{Example: The spherical pendulum hitting a cylindrical surface}

Let us continue with the example of \S\ref{ex:pendulum}. This system is invariant by rotations about the vertical axis, that is, the Lagrangian \eqref{eq:lagrangianpendulum} and the constraint distribution \eqref{eq:constraintpendulum} are invariant by the (tangent lift of) following action:
${\rm SO}(2)\times Q\to Q,\quad(\psi,(\theta,\varphi))\mapsto(\theta,\varphi+\psi)$.

Note that this action is free and proper except when $\theta=k\pi$ for some $k\in\mathbb Z$. Hence, the following is only valid for trajectories not passing through those configurations. The reduced configuration space and its boundary are given by 
\begin{equation*}
\Sigma=Q/{\rm SO}(2)\simeq\{\theta\in S^1\mid L\sin\theta\leq 1\},\qquad\partial\Sigma=\{\theta\in S^1\mid L\sin\theta= 1\}.
\end{equation*}
Hence, the tangent bundle and its annihilator are given by
$T\partial\Sigma=0,\quad (T\partial\Sigma)^\circ=\operatorname{span}\{d\theta\}=T^*\partial\Sigma$. Since we are working locally, $\tilde{\mathfrak{so}}(2)=\Sigma\times\mathfrak{so}(2)=\Sigma\times\mathbb R$ and we may choose the trivial principal connection on $Q\to\Sigma$, whence
\begin{equation*}
(TQ)/{\rm SO}(2)\simeq T\Sigma\oplus\tilde{\mathfrak{so}}(2)=T\Sigma\times\mathbb R,\quad[\theta,\varphi;v_\theta,v_\varphi]_{{\rm SO}(2)}\mapsto(\theta,v_\theta;\xi=v_\varphi).
\end{equation*}
The reduced Lagrangian thus reads
\begin{equation*}
\ell(\theta,v_\theta,\xi=v_\varphi)=\frac{1}{2}mL^2(v_\theta^2+\xi^2\sin^2\theta)-mgL\cos\theta,
\end{equation*}
and analogous for the reduced constraint distributions:
\begin{equation*}
\Delta_\Sigma=\operatorname{span}\{\partial_\theta\}=T\Sigma,\qquad\Delta_{\tilde{\mathfrak{so}}(2)}=\operatorname{span}\{f(\theta)\partial_\varphi\}.
\end{equation*}
The corresponding annihilators vanish $\Delta_\Sigma^\circ=0$ and $\Delta_{\tilde{\mathfrak{so}}(2)}^\circ=0$.

By gathering all, the implicit Lagrange--Poincaré equations \eqref{eq:LPeqs} for a curve
\begin{equation*}
(\theta,v_\theta,p_\theta,\xi=v_\varphi,\rho=p_\varphi,\zeta):[t_0,t_1]\to(T\Sigma\oplus T^*\Sigma)\times\mathbb R\times\mathbb R\times\mathbb R,
\end{equation*}
read
\begin{equation*}
\begin{array}{l}
    \left\{\begin{array}{l}
    \displaystyle \dot p_\theta=mL\sin\theta(L\xi^2\cos\theta+g),\vspace{2mm}\\
    \displaystyle\dot\theta=v_\theta,\hspace{22mm} p_\theta=mL^2v_\theta,
    \end{array}\right.\vspace{2mm}\\
    \left\{\begin{array}{l}
    \displaystyle\dot\rho=0,\vspace{2mm}\\
    \displaystyle\zeta=\xi=\mu_0 f(\theta),\qquad \rho=mL^2\xi\sin^2\theta.
    \end{array}\right.
\end{array}
\end{equation*}
on $[t_0,t_1]-\left\{\tilde t\right\}$, where $\mu_0\in\mathbb R$ is the Lagrange multiplier. The impact condition at $t=\tilde t$ given in \eqref{eq:reducedimpactL} reads
$p_\theta^+-p_\theta^-=\lambda_0$, where we denote $p_\theta^+=p_\theta\left(\tilde t^+\right)$, etc., and $\lambda_0\in\mathbb R$ is the Lagrange multiplier. Similarly, the condition \eqref{eq:reducedimpactDelta} reads
$\dot\theta^+=v_\theta^+,\qquad\zeta^+=\xi^+=\lambda_0 f(\theta)$, where $\dot\theta^+=\theta\left(\tilde t^+\right)$, etc., and $\lambda_0$ is the Lagrange multiplier.

\section{Conclusions and future work}

In this paper, the Lagrange--d'Alembert--Pontryagin action functional has been extended to the nonsmooth setting to account for nonholonomic systems undergoing elastic collisions. The configuration space of these systems is a smooth manifold with boundary and the impact takes place when the trajectory of the system reaches the boundary. The dynamical equations thus obtained are known as the implicit Euler--Lagrange equations with collisions and naturally include the energy conservation, as well as the appropriate conditions for the impact. Furthermore, for systems with symmetries, the geometric structures describing the system are reduced, yielding the nonholonomic implicit Lagrange--Poincaré equations with collisions, which consist of the horizontal equations, the vertical equations and the conditions for the impact. Both in the original and reduced formulations, the equations are first obtained by applying the variational principle with an auxiliary parameter $\tau$ and, then, they are reparametrized to rewrite them in terms of the time $t$. Lastly, the theory is illustrated through some examples.

For future work, we would like to explore the following lines:
\begin{enumerate}
    \item \textbf{Symmetry breaking}: When the system is not invariant by the action of the entire Lie group, but only by its restriction to a subgroup. This situation, which is particularly relevant in fluid dynamics \cite{gay2010reduction}, is related to systems whose symmetry group is a semi-direct product and leads to advected parameters in the reduced equations \cite{GaYo2015}.
    \item \textbf{Variational integrators}: A discrete counterpart of this theory would be highly desirable in order to obtain numerical schemes that preserve the geometric structures underlying the dynamical equations, thus leading to well-behaved simulations even for long times \cite{FeMaOrWe2003,RoLe2023}.
    \item \textbf{Interconnection}: Interconnection for Lagrange--Dirac systems \cite{JaYo2014,Ro2023} may be extended to the nonsmooth setting. The collisions are expected to be transferred due to the coupling, yielding impacts on systems that originally have no collisions.
    \item \textbf{Nonelastic impacts}: This situation may be modelled by introducing a restitution coefficient in the reset map \cite{brogliato1999nonsmooth} and will allow for treating a wider number of physical systems.
\end{enumerate}

\section*{Funding.}
ARA has been partially supported by Ministerio de Ciencia e Innovación (Spain) under grant PID2021-126124NB-I00. LC acknowledges financial support from Grant PID2022-137909NB-C21 funded by MCIN/AEI/ 10.13039/501100011033.

\bibliographystyle{siamplain}
\bibliography{references}

\end{document}